\documentclass{article}

\usepackage{microtype}
\usepackage{graphicx}
\usepackage{subcaption}
\usepackage{booktabs} 
\usepackage{multirow}
\usepackage{makecell}

\usepackage{hyperref}

\usepackage[preprint]{icml2026}

\usepackage{amsmath,amssymb,amsthm,mathtools,bbm}
\usepackage{amsfonts}
\usepackage{pifont}
\usepackage{algorithm}
\usepackage{algorithmic}
\usepackage{enumitem}  
\usepackage{tcolorbox}

\usepackage{hyperref}
\usepackage[capitalize, noabbrev]{cleveref}
\crefname{section}{Sec.}{Secs.}
\Crefname{section}{Section}{Sections}
\Crefname{table}{Table}{Tables}
\crefname{table}{Tab.}{Tabs.}
\Crefname{appsec}{Appendix}{Appendices}
\crefname{appsec}{Appendix}{Appendices}
\crefname{theorem}{Thm.}{Thms.}
\Crefname{theorem}{Theorem}{Theorems}
\Crefname{equation}{Equation}{Equations}
\crefname{equation}{Eq.}{Eqs.}

\newtheorem{theorem}{Theorem}
\newtheorem{lemma}{Lemma}

\newtheorem{corollary}{Corollary}
\theoremstyle{definition}
\newtheorem{definition}{Definition}
\theoremstyle{remark}
\newtheorem{remark}{Remark}

\newcommand{\KL}{\mathrm{KL}}
\newcommand{\E}{\mathbb{E}}
\newcommand{\1}{\mathbbm{1}}

\usepackage{xcolor} 

\newcommand{\topic}[1]{\noindent\textbf{#1}}

\usepackage[textsize=tiny]{todonotes}

\icmltitlerunning{}

\begin{document}

\twocolumn[
  \icmltitle{\emph{MemPot}: Defending Against Memory Extraction Attack with \\ Optimized Honeypots}



  \icmlsetsymbol{equal}{*}

  \begin{icmlauthorlist}
    \icmlauthor{Yuhao Wang*}{nus}
    \icmlauthor{Shengfang Zhai*}{nus}
    \icmlauthor{Guanghao Jin}{sustech}
    \icmlauthor{Yinpeng Dong}{thu}
    \icmlauthor{Linyi Yang}{sustech}
    \icmlauthor{Jiaheng Zhang}{nus}
  \end{icmlauthorlist}

  \icmlaffiliation{nus}{National University of Singapore}
  \icmlaffiliation{thu}{Tsinghua University}
  \icmlaffiliation{sustech}{Southern University of Science and technology}

  \icmlcorrespondingauthor{Shengfang Zhai}{shengfang.zhai@gmail.com}
  \icmlcorrespondingauthor{Linyi Yang}{yangly6@sustech.edu.cn}

  \icmlkeywords{Machine Learning, ICML}

  \vskip 0.3in
]



\printAffiliationsAndNotice{\icmlEqualContribution}

\begin{abstract}
Large Language Model (LLM)-based agents employ external and internal memory systems to handle complex, goal-oriented tasks, yet this exposes them to severe extraction attacks, and corresponding defenses are currently lacking.
In this paper, we propose \emph{MemPot}, the first theoretically verified defense framework against memory extraction attacks by injecting optimized honeypots into the memory. 
Through a two-stage optimization process, \emph{MemPot} generates trap documents that maximize the retrieval probability for attackers while remaining inconspicuous to benign users. 
We model the detection process as Wald’s Sequential Probability Ratio Test (SPRT) and theoretically prove that \emph{MemPot} achieves a lower average number of sampling rounds compared to optimal static detectors. 
Empirically, \emph{MemPot} significantly outperforms state-of-the-art baselines, achieving a 50\% improvement in detection AUROC and an 80\% increase in True Positive Rate under low False Positive Rate constraints. 
Furthermore, our experiments confirm that \emph{MemPot} incurs zero online inference latency and preserves the agent's utility on standard tasks, verifying its superiority in safety, harmlessness and efficiency.
\end{abstract}

\section{Introduction}
Large language model (LLM) is now becoming one of the most important AI technologies in daily life with its impressive performance~\citep{GPT-4,llm_survey_zhao_23}.
Building on recent advances in LLMs~\citep{achiam2023gpt,liu2024deepseek,grattafiori2024llama}, LLM-based agents are equipped with additional functionalities to perform complex, goal-oriented tasks~\citep{xi2023rise_LLMagent_survey}. A typical agent follows a structured pipeline that processes user instructions, gathers environmental information, retrieves relevant knowledge and past experiences, formulates action plans, and executes them in the environment~\cite{agent_survey_wang_24,hu2025memory}. This paradigm has enabled diverse real-world applications, including healthcare~\cite{health_agent_abbasian_24}, autonomous driving~\cite{AgentDriver_mao_colm24}, finance~\cite{ding2024largelanguagemodelagent}, code generation ~\cite{hong2024metagpt}, business management~\citep{salesforce_agentforce} and web interaction~\cite{ webshop_yao_neurips22,ReAct_yao_ICLR23}, positioning LLM agents as a central AI technology today. Typically, an agent is equipped with an external memory, which usually contains domain-specific knowledge, and an internal memory, where past experiences and user-interaction histories are stored.

\begin{figure}[t]
    \centering
    \includegraphics[width=0.8\linewidth, trim=40 280 10 10, clip]{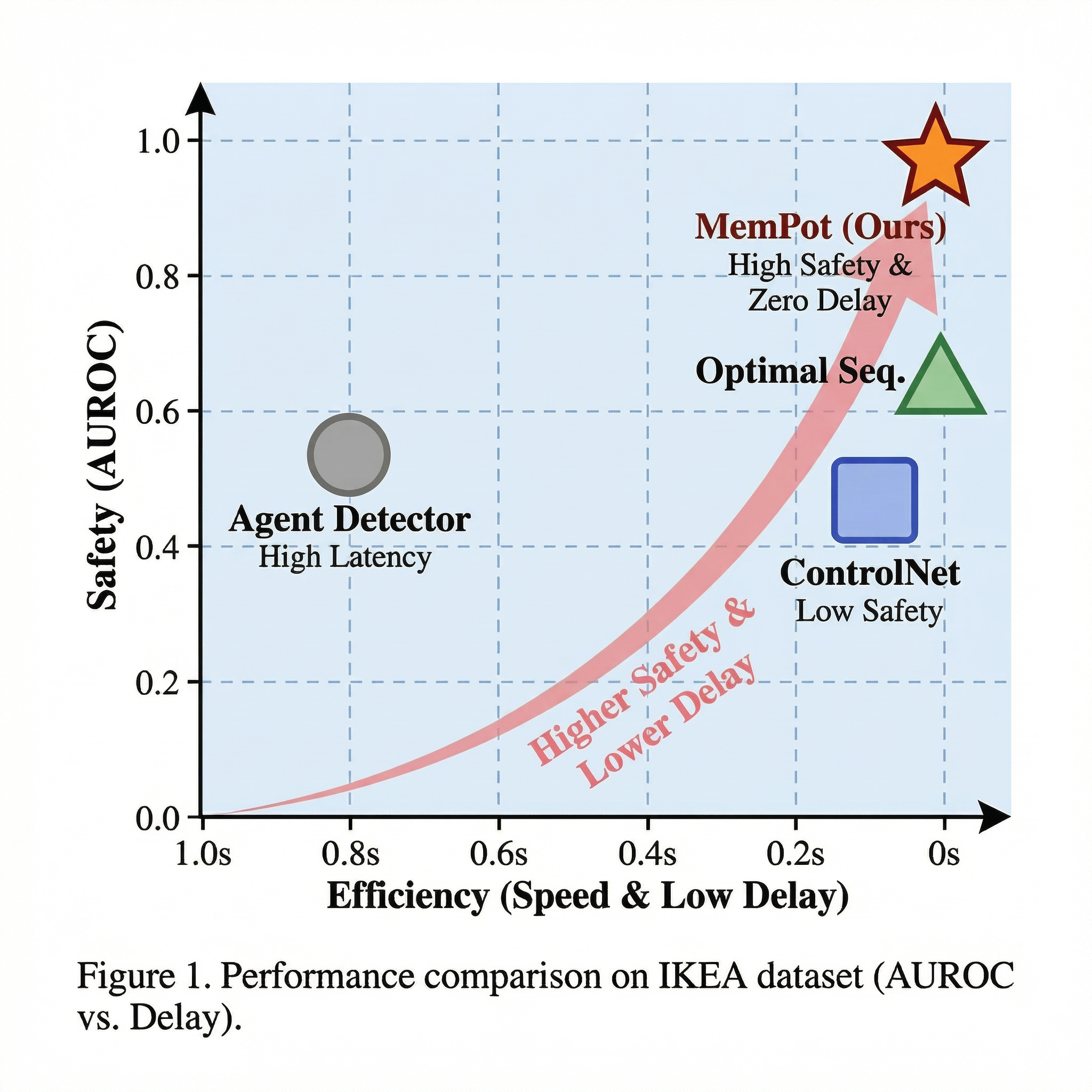}
    \caption{Performance comparison of \emph{MemPot} and existing methods (AUROC vs. Delay).}
    \label{fig:method-comparison}
    \vskip -2em
\end{figure}

Despite their rapidly increasing deployment, LLM-based agents pose serious risks of privacy and knowledge leakage. Modern LLM-based agents frequently retrieve information from external memory that contains private and high-value domain data~\cite{health_agent_abbasian_24,domain_chatbots_kulkarni_24,salesforce_agentforce}. While such retrieval improves task performance, it also introduces significant security risks. For example, the ForcedLeak vulnerability in Salesforce Agentforce enabled large-scale exfiltration of Customer Relationship Management (CRM) data~\cite{noma_forcedleak_agentforce}, underscoring the vulnerability of external memory to data leakage in real-world agent systems. In addition to external sources, agents also maintain internal memory modules that store long-term interaction histories, including past user instructions and agent-generated responses~\cite{survey_agent_memory}. Because these internal memories inherently contain sensitive user data, such as personal preferences and private records, their leakage can lead to serious privacy violations, such as exposure of medical information and purchase history~\cite{wang2025unveiling}.

Prior works has explored extraction attacks against external knowledge database~\cite{wang2025IKEA,jiang2024rag,cohen2024unleashing,zeng2024good,qi2025spill}, which can be applied to Retrieval
Augmented Generation (RAG) system~\cite{RAG_survey_Fan_KDD24} and agents' external memory~\cite{hu2025memory}. Recent works also discovered privacy attack on agents' internal memory, such as long-term interaction histories~\cite{wang2025unveiling}. 

Despite recent progress in defending extraction attacks~\cite{zhang2024intention,zeng2025sage,agarwal2024prompt,jiang2024rag,yao2025controlnet}, existing defense methods still have limitations. 
Most current defenses focus on per-query detection and rely on real-time inference with large language models or auxiliary detectors~\cite{zhang2024intention,zeng2024autodefense,yao2025controlnet}. As a result, they struggle to identify stealthy extraction attacks that employ benign-looking queries and gradual interaction patterns~\cite{wang2025IKEA, wang2025unveiling, jiang2024rag, cohen2024unleashing}, as they lack mechanisms to aggregate evidence across multiple retrieval steps. 
Moreover, their reliance on real-time inference introduces inference latency, which greatly undermines the interactive smoothness with users, limiting their practical applications.

To address the limitations of prior defenses, we propose \emph{MemPot}, a \textbf{zero-online-cost} extraction defense framework with theoretical guarantees. 
\emph{MemPot}
inserts optimized honeypot documents into the memory, and these honeypot documents are designed to attract attackers while remaining inconspicuous to benign users.
The main challenge of \emph{MemPot} is to ensure the quality of service (QoS) for normal users while maximizing the detection performance against attackers. 
The challenges thus involve: 
\textbf{(1)} This requires the honeypot documents to be sufficiently attractive to attackers while remaining inconspicuous to normal users. 
\textbf{(2)} The honeypot documents must be harmless and not mislead normal users, which necessitates careful design to avoid negative impacts on user experience. 
\textbf{(3)} To limit the impact of attackers, it is important to minimize the rounds of detection, as earlier detection can prevent further leakage of private information.

Our approach addresses these challenges through a two-stage optimization strategy to balance detection efficiency with Quality of Service (QoS). In the first stage, we optimize honeypot embeddings using contrastive loss to maximize the statistical separability between attacker and normal user retrieval patterns. By formalizing this detection task as a sequential hypothesis testing problem using Wald's Sequential Probability Ratio Test (SPRT), we theoretically prove that this optimization objective leads to minimized average detection rounds, surpassing the efficiency limits of any optimal static detector. In the second stage, we address safety requirements by inverting these optimized embeddings into concrete, benign documents, ensuring they remain harmless and do not mislead normal users. Empirically, \emph{MemPot} validates these theoretical guarantees, achieving a 50\% improvement in detection AUROC and an 80\% increase in TPR@1\%FPR over state-of-the-art baselines. Furthermore, our results confirm that \emph{MemPot} maintains near-zero detection delay and negligible impact on benign user utility, demonstrating its superiority in both defense efficiency and practical utility.
In summary, our main contributions are:

\begin{itemize}
    \item We propose \emph{MemPot}, the first general defense framework against memory extraction attacks. By employing a novel two-stage optimization strategy, \emph{MemPot} inserts harmless honeypots into the memory, and performs sequential detection based on accumulated retrieval evidence, achieving notable performance without affecting the Quality of Service (QoS).

    \item We formulate detection as a sequential hypothesis testing problem and apply Wald’s SPRT to construct an optimal detector, theoretically minimizing the expected detection rounds and outperforming static detectors without honeypots.
    
    \item Extensive experiments across two datasets and two agent settings show that \emph{MemPot} consistently achieves near-perfect detection accuracy against state-of-the-art extraction attacks with zero-online latency, and \emph{MemPot} have negligible impact on agent utility.

\end{itemize}

\section{Related Works}
\subsection{LLM Agents}
Large Language Models (LLMs) have demonstrated revolutionary capabilities in language understanding, reasoning, and generation~\citep{llm_survey_zhao_23}. 
Building on these advances, LLM agents use LLMs and supplement with additional functionalities to perform more complex tasks \citep{xi2023rise_LLMagent_survey}. 
Its typical pipeline consists of the following key steps: taking user instruction, gathering environment information, retrieving relevant knowledge and past experiences, giving an action solution based on the above information, and finally executing the solution \cite{agent_survey_wang_24}. This pipeline enables agents to support various real-world applications, such as healthcare \cite{health_agent_abbasian_24}, web applications \cite{webshop_yao_neurips22}, and autonomous driving \cite{AgentDriver_mao_colm24}. 

\begin{table}[t]
\footnotesize
  \centering
  \caption{Comparison of defense methods (details in~\cref{sec:setups}). \textbf{MemPot} achieves the best performance with zero online cost.
  }
  \label{tab:defense comparison overview}
  \resizebox{\columnwidth}{!}{
  \begin{tabular}{lcccc}
    \toprule
    \makecell{Defense\\Method} & \makecell{Distribution\\Change} & \makecell{Detection\\Paradigm} & \makecell{Detection\\Performance} & \makecell{Online\\Cost} \\
    \midrule
    ControlNet         & $\times$     & Single Turn & Low    & Middle \\
    Agent              & $\times$     & Single Turn & Low & High   \\
    Optimal Seq & $\times$     & Sequential  & Middle & Middle \\
    \textbf{MemPot}    & \checkmark   & \textbf{Sequential}  & \textbf{High}   & \textbf{Zero}   \\
    \bottomrule
  \end{tabular}
  }
  \vskip -1.5em
\end{table}

\subsection{Privacy Risk in Memory System}
The private information of an LLM agent mainly originates from two sources: 
(1) In external memory domain, agents usually employ RAG to retrieve high-value domain-specific records (e.g., patient prescriptions~\cite{chatdoctor_li_23}) to enhance generation~\cite{hu2025memory, RAG_Lewis_neurips20,domain_chatbots_kulkarni_24}. 
(2) In internal memory domain, the memory module emerges as a new risk source by archiving sensitive user-agent interactions, specifically pairs of private instructions and agent solutions~\cite{survey_agent_memory}.
While prior research has demonstrated data leakage risks in RAG systems through various extraction attacks~\cite{zeng2024good,jiang2024rag,di2024pirates,cohen2024unleashing,wang2025IKEA}, recent studies have further confirmed the tangible threat of extracting sensitive details directly from the agent's internal memory~\cite{wang2025unveiling}. Hence, it is urgent to explore effective and fundamental defense strategy to mitigate such attacks.

\subsection{Defense against Extraction Attack on Memory System}
Current defense strategies primarily fall into two categories: embedding-level detection and text-level detection. As a representative of embedding-level approaches, ControlNet~\cite{yao2025controlnet} measure current query's distributional shift between benign query embeddings to identify potential extraction attacks. In contrast, text-level detection typically relies on Large Language Models (LLMs) to discern query intentions or employs multi-agent systems to analyze the potential impact of queries~\cite{zhang2024intention,zeng2024autodefense,agarwal2024prompt}.
While these works have made significant progress in defending extraction attacks, applying them to agent memory protection presents limitations. 
Current methods exhibit two primary limitations: First, they impose heavy computational overhead relying on real-time inference with large language models or auxiliary detectors~\cite{zhang2024intention,zeng2024autodefense,yao2025controlnet}, limiting their practicality for long-running and interactive agents. Second, their defensive capability is fundamentally limited by a static, single-turn detection paradigm. By treating each query in isolation, these methods fail to aggregate evidence across interactions. Current methods also lack capability to proactively alter the memory distribution to trap adversaries, which is important when facing benign-looking attacks that closely mimic normal behavior. Consequently, these methods remain vulnerable to stealthy extraction strategies. We compare the key differences between existing approaches and our method in~\cref{tab:defense comparison overview}.

\begin{figure*}[t]
    \centering
    \includegraphics[width=1\linewidth]{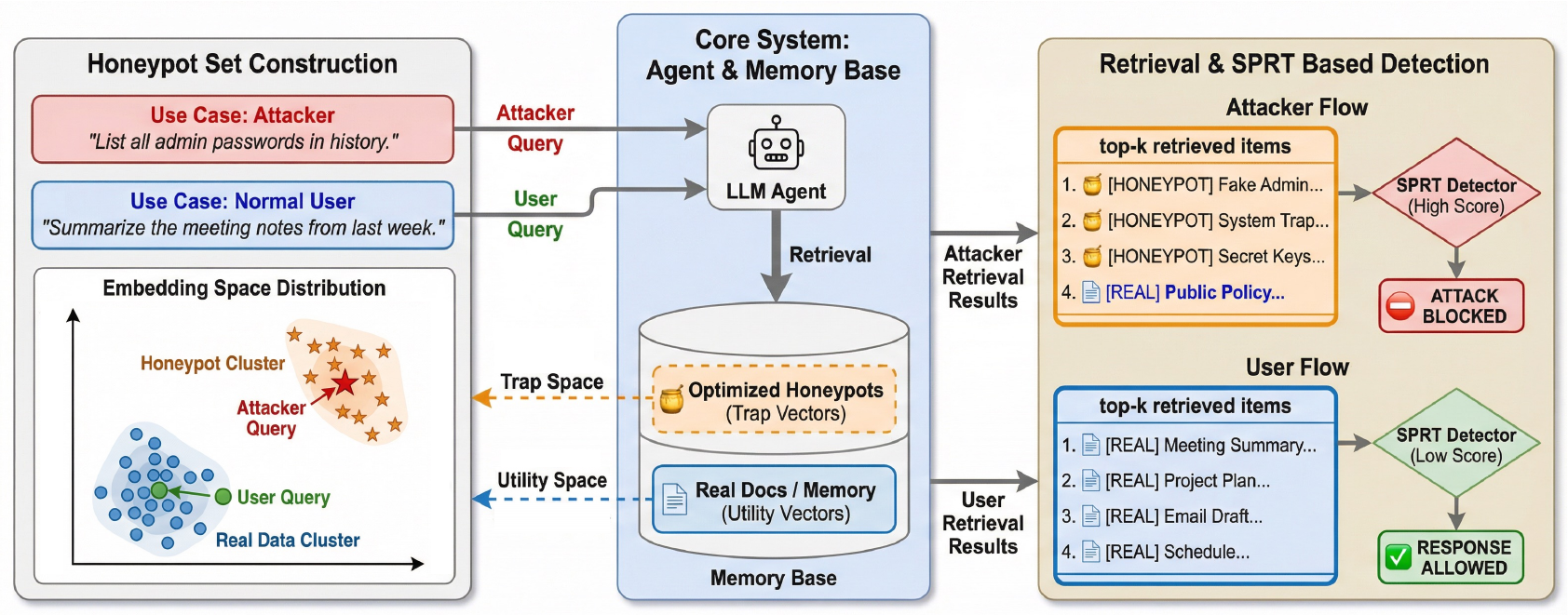}
    \caption{Overview of \emph{MemPot} Detection Framework.}
    \label{fig:main}
    \vskip -1em
\end{figure*}

\section{Preliminary}
\subsection{Threat Model}
\noindent\textbf{Defense Scenario. }
We consider a LLM-based agent service provider as the defender, who maintains a memory $\mathcal{D}$ and provides services to users. The defender aims to protect the privacy of the memory from potential attackers while ensuring high-quality service for normal users. We assume that attackers have the same access privileges as normal users and have no prior knowledge of the defense setting.

\noindent\textbf{Defender's Prior knowledge.} 
The defender is assumed to has all knowledge about the agent system, including the retriever, LLM, the content and index embeddings of the memory.
We assume the defender only have partial knowledge about the attacker, which means that the defender has access to a small set of attacker queries $\mathrm{Q}_a = \{q_1^{(1)}, q_2^{(1)}, \dots, q_M^{(1)}\}$, which can be collected from historical attack logs.

\noindent\textbf{Defender Goal. }
The defender aims to accurately detect attackers while minimizing the impact on normal users. The goal can be summarized as two parts:
(1) Detection Accuracy and Efficiency: The defender aims to maximize the detection accuracy while minimizing the average detection rounds, thereby reducing the potential leakage of private information.
(2) Quality of Service (QoS): The defender aims to ensure that the presence of honeypot documents does not significantly degrade the user experience for normal users, which can be measured by the false positive rate (FPR) of detection and empirical utility experiments.

\subsection{Sequential Hypothesis Testing Model}
\label{sec:sprt_model}
In this part, we formalize attacker detection with \emph{optimized Honeypots} as a sequential hypothesis testing problem. Let \(d\in\mathbb{N}\). The fixed document embeddings are \(\mathcal{E}_{\mathrm{doc}}=\{e_i\}_{i=1}^N\subset\mathbb{R}^d\). Trainable honeypot embeddings are \(\mathcal{E}_{\mathrm{pot}}(\theta)=\{u_j(\theta)\}_{j=1}^P\subset\mathbb{R}^d\), and the augmented database is:
\[
\mathcal{E}_{\mathrm{aug}}(\theta)=\mathcal{E}_{\mathrm{doc}}\cup\mathcal{E}_{\mathrm{pot}}(\theta).
\]

\noindent\textbf{Sequential Testing Model. }There are two query sources:
\begin{equation}
    \begin{aligned}
        &q\sim \mathcal{Q}_1 \quad \text{(attacker, hypothesis \(H_1\))},\\
        &q\sim \mathcal{Q}_0 \quad \text{(normal, hypothesis \(H_0\))}.
    \end{aligned}
\end{equation}

At round \(t\), we observe
$
O_t=\Phi(q_t;\theta),
$
where \(\Phi\) deterministically maps the query and the augmented index to the retrieval information (e.g., query,  returned indices and similarity scores). Let
$
f_{1,\theta}\ \text{and}\ f_{0,\theta}
$
denote the laws of a single-round observation \(O\) under \(H_1\) and \(H_0\), respectively. 
Given type-I/II error budgets \((\alpha,\beta)\in(0,1)^2\), our \emph{objective} is to design a sequential test with minimum stopping rounds \(N\) and terminal decision \(D_N\in\{H_0,H_1\}\) such that
\begin{equation} \label{eq:op-goal}
    \begin{aligned}
        \min \big\{ E_1[N],\,& E_0[N] \big\},
        \\
        \text{s.t.}\quad
        P_0(D_N=H_1)\le \alpha,& \quad
        P_1(D_N=H_0)\le \beta
    \end{aligned}
\end{equation}
where \(E_i[\cdot]\) denotes expectation under hypothesis \(H_i\). 

\noindent\textbf{Wald’s approximated SPRT. }
We utilize Wald’s approximated SPRT~\cite{wald1992sequential} to solve the optimization problem in~\ref{eq:op-goal}. 
Define the per-round log likelihood ratio (LLR) and accumulated log likelihood ratio
\begin{equation}\label{eq:llr}
\ell_\theta(O)=\log\frac{f_{1,\theta}(O)}{f_{0,\theta}(O)},\qquad
S_n=\sum_{t=1}^n \ell_\theta(O_t).
\end{equation}
Define information drift under two hypothesis:
\begin{equation}\label{eq:drifts}
\begin{aligned}
    &\mu_1(\theta)=\E_{f_{1,\theta}}[\ell_\theta(O)]
    =\KL(f_{1,\theta}\|f_{0,\theta}),
    \\
    &\mu_0(\theta)=\E_{f_{0,\theta}}[\ell_\theta(O)]
    =-\KL(f_{0,\theta}\|f_{1,\theta}).
\end{aligned}
\end{equation}
\citet{wald1992sequential} shows that for SPRT with boundaries \(A=\frac{1-\beta}{\alpha}\) and \(B=\frac{\beta}{1-\alpha}\) stops at
\begin{equation}\label{eq:sprt_boundry}
    N=\inf\{n:\ S_n\ge \log A\ \text{ or }\ S_n\le \log B\}.
\end{equation}
and at least one of the two errors will be controlled at the desired level:
\begin{align}
\mathbb{P}_0(D_N=H_1)+\mathbb{P}_1(D_N=H_0)\leq \alpha + \beta. \label{eq:inequality_Wald}
\end{align}
More precisely, the expected average sampling number (ASN)
\footnote{In this paper, we assume the observations are with Markov property conditional on the hypothesis (Proofs of ASN approximation for SPRT with Markov observations are detailed in Appendix.~\ref{sec:prelim},~\cref{lem:ASN_Markov}).}
in two hypothesis with negligible overshoot can be approximated by:
\begin{equation}\label{eq:wald-approx}
E_1[N] \approx \frac{|\log B|}{\mu_1(\theta)},
\qquad
E_0[N] \approx \frac{|\log A|}{|\mu_0(\theta)|}.
\end{equation}
Hence, for fixed \((\alpha,\beta)\), increasing \(\mu_1(\theta)\) and \(|\mu_0(\theta)|\) decreases the expected sample sizes that attain those error budgets.

\section{\emph{MemPot}: Optimization and Detection}
Our methodology is grounded in the insight that memory retrieval is an iterative and sequential process. Unlike computationally expensive per-query detection, we leverage the retrieval mechanism itself as a zero-cost indicator. We aggregate evidence across the interaction trajectory to distinguish attackers from benign users. 
To amplify these discriminative signals, \emph{MemPot} injects optimized honeypots designed to stimulate adversarial behavior without disrupting normal service.
In this section, we detail our framework (\cref{fig:main}):
We derive the optimization of vector-form honeypots based on SPRT theory and prove its optimality in minimizing detection rounds in \cref{sec:vecpot}. 
We then convert these vectors into harmless text-form documents via safety-constrained embedding inversion in \cref{sec:textpot}. 
We finally present practical Log-Likelihood Ratio (LLR) estimation methods to execute the sequential detection in \cref{sec:pot_detect}.

\subsection{Honeypot Vector in Semantic Embedding Space}
\label{sec:vecpot}

The trainable honeypot parameters \(\theta\) influence the observation distributions \((f_{1,\theta},f_{0,\theta})\) through retrieval mechanism, and thus determine both error probabilities and sampling efficiency. We will leverage this dependence to derive an optimization objective for vector-form honeypots that increases the statistical separability of \(f_{1,\theta}\) and \(f_{0,\theta}\) under the constraints of error budgets in \cref{eq:op-goal}.
\begin{theorem}[InfoNCE upper-bound by information drift, Proof in Appendix.~\ref{pf:infoNCE}]
\label{thm:infoNCE upb short}
Draw index $j\sim\mathrm{Unif}\{1,\dots,K\}$, then
$q_j\sim Q_1$ and $(q_i)_{i\neq j}\sim Q_0$ independently of $j$. 
For any score function $h: O\to\mathbb R$, define the InfoNCE loss
$$
\mathcal L_{\mathrm{NCE},K}(h;\theta)
:=-\,\E\!\left[\log
\frac{e^{h(\Phi(q_j);\theta)}}{\sum_{i=1}^{K} e^{h(\Phi(q_i);\theta)}}
\right].
$$
Then, for every $K\ge2$,
\begin{equation}\label{eq:nce-kl}
-\mathcal L_{\mathrm{NCE},K}(h;\theta)\ \le\ \mu_1(\theta)-\text{log}(K).
\end{equation}
\end{theorem}

\noindent\cref{thm:infoNCE upb short} shows that decreasing $\mathcal L_{\mathrm{NCE},K}(h;\theta)$ is equivalent to improve $\mu_1(\theta)$'s lower bound and therefore decreasing upper bound of $E_1[N]$ with~\cref{eq:wald-approx}.

In real scenario, memory system usually only returns the top-$k$ similar entries. We therefore define a top-$k$ masked similarity score here for tighter bound.
Let \(\mathcal{P}\) be the set of honeypots $\mathcal{E}_{\text{pot}}(\theta)$'s indices. Define a per-query score
\begin{equation}\label{eq:gk}
g_k(q;\theta):= \frac{1}{k} \sum_{j\in \mathcal{T}_k(q) \cap \mathcal{P}} s(q,E_j),
\end{equation}
where cosine similarity function $s(\mathbf{u},\mathbf{v}) = \frac{\mathbf{u}^\top \mathbf{v}}{|\mathbf{u}|\cdot|\mathbf{v}|}$, top-$k$ index set for query $\mathcal{T}_k(q)$ returns the indices of the $k$ largest elements of cosine similarity in $\mathcal{E}_{\text{aug}}(\theta)$.
Taking $g_k$ as $h$, we then get the final honeypot training loss
\begin{equation}
   \mathcal L_{\text{pot}}(\theta):= \mathcal L_{\mathrm{NCE},K}(g_k;\theta) + \beta\cdot\mathcal{L}_{\text{div}}(\theta),
\end{equation}
where $\beta$ is regularization parameter and $\mathcal{L}_{\text{div}}(\theta)$ is honeypot diversity loss defined by
\begin{equation}
    \mathcal{L}_{\text{div}}(\theta):= \sum_{\substack{1 \le i < j \le P\\u_i,u_j\in\mathcal E_{\text{pot}}}} \frac{2}{P(P-1)} s(u_i,u_j).
\end{equation}

\begin{theorem}[Advantage over static test, Proof in Appendix.~\ref{pf:adv_stat}]
\label{thm:adv_stat short}
Define \(\theta^\star\) the parameter obtained by minimizing \(\mathcal{L}_{\mathrm{NCE}}\), then
for any possibly static fixed-length test achieving \((\alpha,\beta)\) (i.e., tests without honeypot augment), the stopping time $N$ of static test and SPRT with \(\theta^\star\) satisfies:
\begin{equation}
    E_b[N]_{\mathrm{SPRT},\theta^{\star}}\ \le\ E_b[N]_{\mathrm{any},\text{static}},
\end{equation}
    with hypothesis index $b\in\{0,1\}$.
\end{theorem}
We also prove honeypots-augmented SPRT's comparative advantage over static testing in~\cref{thm:adv_stat short}, which theoretically ensures shorter expected stopping time under fixed error control with proper optimized honeypots.

\begin{figure}[t]
  \centering
  \includegraphics[width=0.8\columnwidth, trim=5 80 5 10, clip]{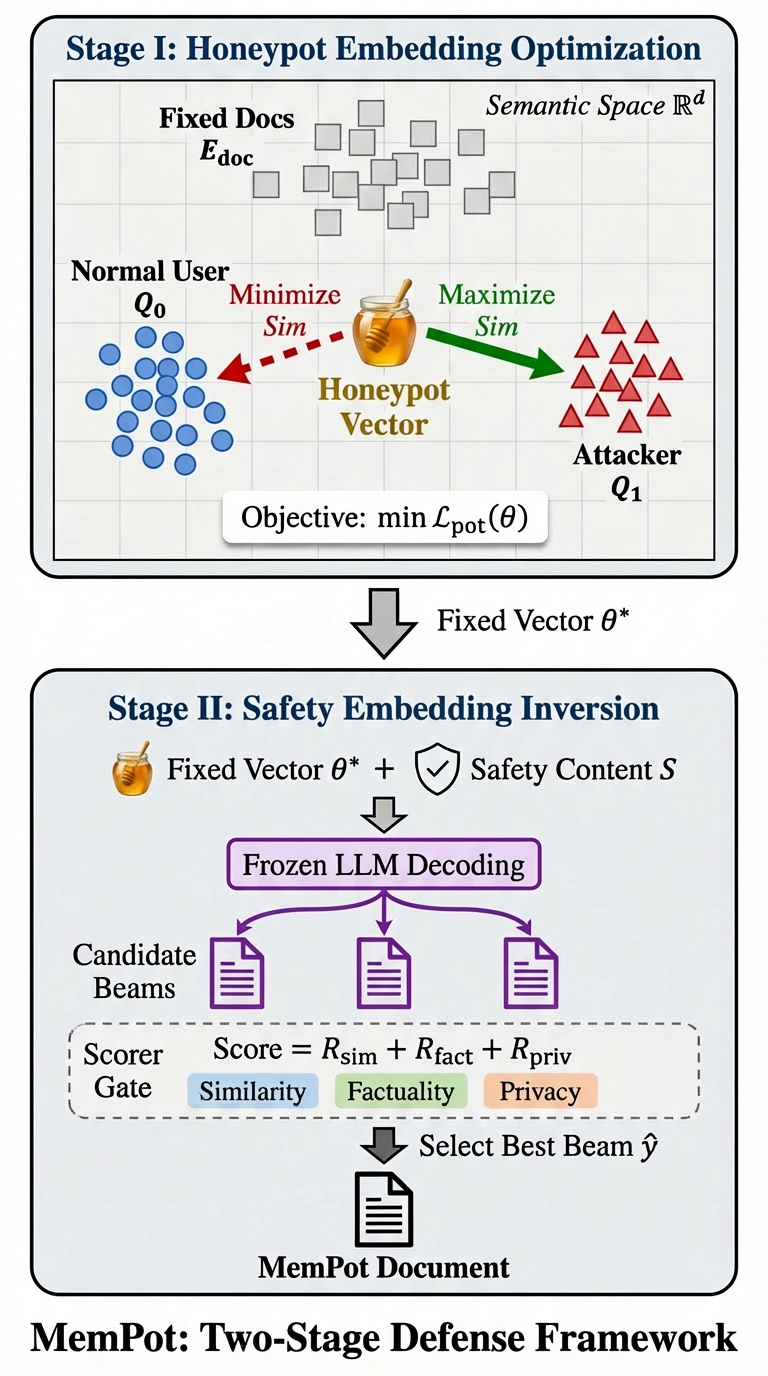}
  \caption{Two Stage Optimization Process of \emph{MemPot}.}
  \label{fig:two-stage}
  \vskip -0.5em
\end{figure}

\subsection{Generate Honeypot Documents from Vectors}
\label{sec:textpot}
To effectively defend against attackers who probe databases using semantical embeddings, honeypots' embeddings must align semantically with their text contents; otherwise, they are unlikely to be retrieved during such exploratory attacks. The honeypot texts must \ding{182} have sentence embeddings sufficiently close to pot vectors optimized in~\cref{sec:vecpot}; \ding{183} maintain factual integrity to minimize the risk of misleading or confusing users; \ding{184} reveal no private information originally contained in the database.
Inspired by~\cite{zhang2025advInv}, we utilize scorer-guided LLM decoding to generate honeypot texts. We use Safety Embedding Inversion to satisfy demands above, which will iteratively search tokens to maximize defined scores with beam-search algorithm.

We design three scorers corresponding to the three demands for honeypot texts and sum them up for overall performance. 
To encourage pot text $y$ to match the given pot vector $e_{\text{p}}$ and be readable, the \textbf{Inversion Scorer} is defined as:
\begin{equation}
    R_{\text{base}}(y,\mathbf e_p)
    =\lambda_{\text{emb}}\cdot s\big(E(y), \mathbf e_{\text{p}}\big)
    \;+\;\lambda_{\mathrm{read}}\cdot \mathrm{readable}(y),
\end{equation}
where $E(\cdot)$ is a sentence embedder with $\mathbf e=E(y)\in\mathbb{R}^d$, $s(\cdot,\cdot)$ is the cosine similarity function and $\mathrm{readable}(y)$ is a readability score defined in~\cite{zhang2025advInv}.
To maintain the factuality of pot texts, we design the \textbf{Factuality Scorer} which requires information in $y$ to be entailed by safety content $S$ (e.g., a topic abstract).
Let $\mathrm{Entail}(y_t\mid S)\in[0,1]$ be the nature language inference entailment (NLI) probability (provided by a pretrained model), the  factuality scorer is defined as:
\begin{equation}
    R_{\text{fact}}(y;S)=\lambda_{\text{fact}}\cdot\mathrm{Entail}(y\mid S).
\end{equation}
To preserve the privacy of the database, we also design a \textbf{Privacy Scorer} to constrain each pot text to have low similarity to the database.
For text $y$ and given document embeddings $\mathcal E_{\mathrm{doc}}$, the scorer is defined as:
\begin{equation}
    R_{\text{priv}}(y|\mathcal E_{\mathrm{doc}})=-\lambda_{\text{priv}}\cdot\max_{j} s\big(E(y),\mathcal E_{\mathrm{doc}}[j]\big).
\end{equation}
For a single honeypot vector, the final optimization objective is
\begin{equation}
\ \ \max_{y\in\mathcal Y}\;
R_{\mathrm{base}}(y,\mathbf e_p)+R_{\text{fact}}(y|S)+R_{\text{priv}}(y|\mathcal E_{\mathrm{doc}}).
\end{equation}
We then use algorithm detailed in~\cref{alg:adv_decoding} to attain the optimized pot texts.

\subsection{Detection with Honeypots}
\label{sec:pot_detect}
To achieve the optimality of SPRT, we need to estimate the log likelihood ratio in~\cref{eq:llr}. When accumulated log likelihood ratio excesses the SPRT boundaries in~\cref{eq:sprt_boundry}, the decision is made and agent system's responses are blocked (e.g.``Unanswerable."). 
Specifically, we designed three ways to approximate the log likelihood ratio.
Let the retrieved document set at step $t$ be $\mathcal{D}_t$, where 
$\mathcal{D}^+_t = \{ d_i \in \mathcal{D}_t \mid \text{is\_pot}(d_i)=1 \}$ denotes \emph{pot} documents 
and $\mathcal{D}^-_t = \{ d_i \in \mathcal{D}_t \mid \text{is\_pot}(d_i)=0 \}$ denotes \emph{non-pot} documents. 
Let $s(d_i)$ denote the similarity score of document $d_i$. The approximation methods are shown as follows:

\topic{1. Pot-NonPot Counts Ratio.}
        We use the ratio between the number of retrieved pot and non-pot documents:
        \begin{equation}\label{eq:cnt_llr}
        \widehat{r}^{\text{cnt}}_t 
        =\frac{|\mathcal{D}^{+}_t| + \varepsilon_{\text{cnt}}}
             {|\mathcal{D}^{-}_t| + \varepsilon_{\text{cnt}}}.
        \end{equation}
        
\topic{2. Pot-NonPot Similarity Ratio.}
        We weight the ratio by similarity scores:
        \begin{equation}\label{eq:sim_llr}
        \widehat{r}^{\text{sim}}_t 
        =
        \frac{ \sum_{d_i \in \mathcal{D}^{+}_t} s(d_i) + \varepsilon_{\text{sim}} }
             { \sum_{d_i \in \mathcal{D}^{-}_t} s(d_i) + \varepsilon_{\text{sim}} }.
        \end{equation}
        
\topic{3. Pot-NonPot Global Similarity Ratio.}
        We further focus on the most relevant evidence by only considering Top-$K$ documents:
        \begin{equation}\label{eq:glb_sim_llr}
        \widehat{r}^{\text{g\_sim}}_t
        =
        \frac{ \frac{1}{K} \sum_{d_i \in \mathcal{D}^{+}_{t,K}} s(d_i) + \varepsilon_{\text{sim}} }
             { \frac{1}{K} \sum_{d_i \in \mathcal{D}^{-}_{t,K}} s(d_i) + \varepsilon_{\text{sim}} },
        \end{equation}
        where $\mathcal{D}^{+}_{t,K}$ means the most similar $K$ honeypot documents, $\mathcal{D}^{-}_{t,K}$ means the most similar $K$ documents that are not honeypots.

The approximated accumulated log likelihood ratio is
$S_t \approx \Lambda_t = \sum_{\tau=1}^{t} \log \widehat{r}_\tau.$
Same as~\cref{eq:sprt_boundry}, the decision is made once the accumulated statistic crosses SPRT boundaries. Notably, the block only happens when terminal decision $D_t=H_1$, which means block happens when $\Lambda_t \ge \log A .$

\subsection{Implementation Details}
\noindent\textbf{Attacker Proxy.}
To approximate the unknown attack distribution $\mathcal{Q}_1$, we employ a neural proxy $A_{\omega}$ that mimics the attacker's behavior at the embedding level. $A_{\omega}$ takes the interaction history $\{o_{1:t-1}\}$ to predict the next query embedding $\hat{e}_{t}$, optimizing a cosine similarity objective to reproduce observed retrieval rankings. This enables the training of defensive strategies against black-box threats without requiring access to their internal algorithms.

\noindent\textbf{User Proxy.}
To approximate the diverse benign distribution $\mathcal{Q}_0$ without extensive real-world data, we leverage LLMs as human simulators. By prompting the LLM with specific intents and retrieval contexts, we synthesize realistic, multi-turn information-seeking trajectories. These generated sequences are then encoded to serve as a robust surrogate for the benign query space. Full description of the proxy building is detailed in Appendix.~\ref{sec:detail-pots-training}.

\begin{table*}[!ht]
\footnotesize
\centering
\caption{Performance of defense methods against external memory extraction attacks on HealthMagicCare and Pokemon datasets.}
\label{tab:external_performance_main}
\resizebox{\textwidth}{!}{
\begin{tabular}{c l cccc cccc}
\toprule
\multirow{3}{*}{Attack} & \multicolumn{1}{c }{\multirow{3}{*}{Defense}} & \multicolumn{4}{c }{HealthMagicCare} & \multicolumn{4}{c}{Pokémon} \\
\cmidrule(lr){3-6} \cmidrule(lr){7-10}
 & & AUROC & TPR@1\%FPR & TPR@10\%FPR & Delay & AUROC & TPR@1\%FPR & TPR@10\%FPR & Delay \\

\midrule
\multirow{3}{*}{RAG-Thief} 
 & ControlNet & 1.00 & 1.00 & 1.00 & 0.06 & 1.00 & \textbf{1.00} & 1.00 & 0.05 \\
 & Agent & 1.00 & \textbf{1.00} & \textbf{1.00} & 0.97 & 0.94 & 0.16 & 1.00 & 0.91 \\
 & MemPot & \textbf{1.00} & 0.96 & 0.99 & \textbf{0} & \textbf{1.00} & 0.98 & \textbf{1.00} & \textbf{0} \\

\midrule
\multirow{3}{*}{DGEA} 
 & ControlNet & 1.00 & 0.95 & 1.00 & 0.06 & 1.00 & 1.00 & 1.00 & 0.05 \\
 & Agent & 1.00 & 1.00 & 1.00 & 0.97 & 0.94 & 0.16 & 1.00 & 0.93 \\
 & MemPot & \textbf{1.00} & \textbf{1.00} & \textbf{1.00} & \textbf{0} & \textbf{1.00} & \textbf{1.00} & \textbf{1.00} & \textbf{0} \\

 \midrule
\multirow{3}{*}{IKEA} 
 & ControlNet & 0.46 & 0 & 0.02 & 0.05 & 0.88 & 0.02 & 0.56 & 0.04 \\
 & Agent & 0.50 & 0 & 0 & 0.98 & 0.20 & 0 & 0.13 & 0.95 \\
 & MemPot & \textbf{0.96} & \textbf{0.42} & \textbf{0.91} & \textbf{0} & \textbf{0.99} & \textbf{0.67} & \textbf{0.98} & \textbf{0} \\
\bottomrule
\end{tabular}
}
\end{table*}

\begin{table*}[t]
\footnotesize
\centering
\caption{Performance of defense methods against internal memory extraction attacks on EHRAgent and RAP web shopping agents.}
\label{tab:internal_performance_main}
\resizebox{\textwidth}{!}{
\begin{tabular}{c l cccc cccc}
\toprule
\multirow{3}{*}{Attack} & \multicolumn{1}{c}{\multirow{3}{*}{Defense}} & \multicolumn{4}{c}{EHRAgent} & \multicolumn{4}{c}{RAP WebShop} \\
\cmidrule(lr){3-6} \cmidrule(lr){7-10}
 & & AUROC & TPR@1\%FPR & TPR@10\%FPR & Delay & AUROC & TPR@1\%FPR & TPR@10\%FPR & Delay \\
\midrule
\multirow{3}{*}{$\text{MEXTRA}_\text{Cosine}$} 
 & ControlNet & 0.59 & 0 & 0.14 & 0.08 & 0.45 & 0 & 0.10 & 0.08 \\
 & Agent & 0.56 & 0 & 0.04 & 0.96 & 0.51 & 0.02 & 0.02 & 0.97 \\
 & MemPot & \textbf{0.99} & \textbf{0.97} & \textbf{1.00} & \textbf{0} & \textbf{1.00} & \textbf{0.94} & \textbf{1.00} & \textbf{0} \\
\midrule
\multirow{3}{*}{$\text{MEXTRA}_\text{Edit}$} 
 & ControlNet & 0.53 & 0.02 & 0.17 & 0.08 & 0.59 & 0.02 & 0.16 & 0.08 \\
 & Agent & 0.67 & 0 & 0.14 & 0.88 & 0.81 & 0.30 & 0.30 & 0.96 \\
 & MemPot & \textbf{0.97} & \textbf{0.86} & \textbf{0.99} & \textbf{0} & \textbf{1.00} & \textbf{1.00} & \textbf{1.00} & \textbf{0} \\
\midrule
\multirow{3}{*}{$\text{MEXTRA}_\text{General}$} 
 & ControlNet & 0.63 & 0.02 & 0.08 & 0.08 & 0.73 & 0 & 0.36 & 0.08 \\
 & Agent & 0.58 & 0 & 0.06 & 0.91 & 0.72 & 0.12 & 0.12 & 0.91 \\
 & MemPot & \textbf{0.94} & \textbf{0.81} & \textbf{0.96} & \textbf{0} & \textbf{1.00} & \textbf{1.00} & \textbf{1.00} & \textbf{0} \\
 \midrule
\multirow{3}{*}{IKEA} 
 & ControlNet & 0.40 & 0 & 0 & 0.10 & 0.71 & 0.06 & 0.34 & 0.10 \\
 & Agent & 0.51 & 0 & 0.02 & 1.01 & 0.47 & 0.02 & 0.02 & 0.96 \\
 & MemPot & \textbf{1.00} & \textbf{0.98} & \textbf{1.00} & \textbf{0} & \textbf{1.00} & \textbf{1.00} & \textbf{1.00} & \textbf{0} \\
\bottomrule
\end{tabular}
}
\end{table*}
\begin{table*}[t]
\footnotesize
\centering
\caption{Comparison between Optimal Sequential Detector and MemPot on external memory extraction attacks.}
\label{tab:health_pokemon_opt}
\resizebox{\textwidth}{!}{
\begin{tabular}{c l ccccc ccccc}
\toprule
\multirow{3}{*}{Attack} & \multicolumn{1}{c}{\multirow{3}{*}{Defense}} & \multicolumn{5}{c}{HealthMagicCare} & \multicolumn{5}{c}{Pokémon} \\
\cmidrule(lr){3-7} \cmidrule(lr){8-12}
 & & AUROC & TPR@1\%FPR & TPR@10\%FPR & Delay & \multicolumn{1}{c}{FDT} & AUROC & TPR@1\%FPR & TPR@10\%FPR & Delay & FDT \\

\midrule
\multirow{2}{*}{RAG-Thief} 
 & Optimal-Seq & 0.97 & 0.69 & 0.89 & 0.04 & \textbf{4} & \textbf{1.00} & 0.95 & \textbf{1.00} & 0.03 & 2 \\
 & MemPot & \textbf{1.00} & \textbf{0.96} & \textbf{0.99} & \textbf{0} & 9 & \textbf{1.00} & \textbf{0.98} & \textbf{1.00} & \textbf{0} & \textbf{1} \\

\midrule
\multirow{2}{*}{DGEA} 
 & Optimal-Seq & 0.90 & 0.39 & 0.64 & 0.04 & 11 & 0.97 & 0.70 & 0.88 & 0.04 & \textbf{1} \\
 & MemPot & \textbf{1.00} & \textbf{1.00} & \textbf{1.00} & \textbf{0} & \textbf{1} & \textbf{1.00} & \textbf{1.00} & \textbf{1.00} & \textbf{0} & \textbf{1} \\

 \midrule
\multirow{2}{*}{IKEA} 
 & Optimal-Seq & 0.70 & 0.11 & 0.27 & 0.03 & 17 & 0.78 & 0.22 & 0.34 & 0.4 & 14 \\
 & MemPot & \textbf{0.96} & \textbf{0.42} & \textbf{0.91} & \textbf{0} & \textbf{9} & \textbf{0.99} & \textbf{0.67} & \textbf{0.98} & \textbf{0} & \textbf{7} \\
 
\bottomrule
\end{tabular}
}

\end{table*}
\begin{table*}[t]
\footnotesize
\centering
\caption{Comparison between Optimal Sequential Detector and MemPot on internal memory extraction attacks.}
\label{tab:ehr_rap_opt}
\resizebox{\textwidth}{!}{
\begin{tabular}{c l ccccc ccccc}
\toprule
\multirow{3}{*}{Attack} & \multicolumn{1}{c}{\multirow{3}{*}{Defense}} & \multicolumn{5}{c}{EHRAgent} & \multicolumn{5}{c}{RAP WebShop} \\
\cmidrule(lr){3-7} \cmidrule(lr){8-12}
 & & AUROC & TPR@1\%FPR & TPR@10\%FPR & Delay & \multicolumn{1}{c}{FDT} & AUROC & TPR@1\%FPR & TPR@10\%FPR & Delay & FDT \\
\midrule
\multirow{2}{*}{$\text{MEXTRA}_\text{Cosine}$} 
 & Optimal-Seq & 0.70 & 0.16 & 0.27 & 0.03 & 22 & 0.68 & 0.06 & 0.22 & 0.04 & 38 \\
 & MemPot & \textbf{0.99} & \textbf{0.97} & \textbf{1.00} & \textbf{0} & \textbf{2} & \textbf{1.00} & \textbf{0.94} & \textbf{1.00} & \textbf{0} & \textbf{3} \\
\midrule
\multirow{2}{*}{$\text{MEXTRA}_\text{Edit}$} 
 & Optimal-Seq & 0.71 & 0.14 & 0.32 & 0.04 & 33 & 0.75 & 0.20 & 0.36 & 0.04 & 27 \\
 & MemPot & \textbf{0.97} & \textbf{0.86} & \textbf{0.99} & \textbf{0} & \textbf{8} & \textbf{1.00} & \textbf{1.00} & \textbf{1.00} & \textbf{0} & \textbf{1} \\
\midrule
\multirow{2}{*}{$\text{MEXTRA}_\text{General}$} 
 & Optimal-Seq & 0.69 & 0.16 & 0.28 & 0.04 & 42 & 0.73 & 0.12 & 0.26 & 0.04 & 35 \\
 & MemPot & \textbf{0.94} & \textbf{0.81} & \textbf{0.96} & \textbf{0} & \textbf{7} & \textbf{1.00} & \textbf{1.00} & \textbf{1.00} & \textbf{0} & \textbf{1} \\
 \midrule
\multirow{2}{*}{IKEA} 
 & Optimal-Seq & 0.78 & 0.20 & 0.36 & 0.05 & 12 & 0.79 & 0.17 & 0.34 & 0.06 & 14 \\
 & MemPot & \textbf{1.00} & \textbf{0.98} & \textbf{1.00} & \textbf{0} & \textbf{2} & \textbf{1.00} & \textbf{1.00} & \textbf{1.00} & \textbf{0} & \textbf{2} \\
\bottomrule
\end{tabular}
}

\end{table*}
\begin{table*}[!ht]
\setlength{\tabcolsep}{9pt}
\centering
\footnotesize
\caption{Comparison of different LLR estimation methods (\cref{sec:pot_detect}) on $\text{MEXTRA}_\text{Cosine}$.}
\label{tab:llr_estimation}
\resizebox{\textwidth}{!}{
\begin{tabular}{lcccccccc}
\toprule
\multirow{2}{*}{Method} & \multicolumn{4}{c}{EHRAgent} & \multicolumn{4}{c}{RAP-web} \\
\cmidrule(r){2-5} \cmidrule(l){6-9} 
 & AUROC & TPR@1\%FPR & TPR@10\%FPR & FDT & AUROC & TPR@1\%FPR & TPR@10\%FPR & FDT \\
\midrule
Count & 0.99 & 0.97 & \textbf{1.00} & 2 & \textbf{1.00} & \textbf{0.94} & \textbf{1.00} & 3 \\
Similarity & 0.74 & 0.65 & 0.75 & 4 & \textbf{1.00} & 0.92 & \textbf{1.00} & 3 \\
Global-Sim & \textbf{1.00} & \textbf{0.98} & \textbf{1.00} & \textbf{1} & \textbf{1.00} & 0.90 & \textbf{1.00} & \textbf{1} \\
\bottomrule
\end{tabular}
}
\end{table*}

\section{Experiments}
\subsection{Setups}\label{sec:setups} 
\textbf{Datasets and Evaluation Setting.} We evaluate \emph{MemPot} on four benchmarks: HealthMagicCare (20k)~\cite{LavitaHealth} and Pokémon (9.46k)~\cite{asoria2024pokemon} for external memory (top-4 retrieval); and EHRAgent~\cite{EHRAgent_shi_emnlp24} and WebShop RAP~\cite{webshop_yao_neurips22, RAP_kagaya_24} for internal memory (300 records, top-3/4 retrieval). We employ DeepSeek-v3.2 (685B)~\cite{liu2024deepseek} as the agent backbone and all-mpnet-base-v2~\cite{song2020mpnet} for sentence embeddings. We use mDeBERTa-v3-base~\cite{he2021debertav3} as NLI model. Honeypots are injected at a fixed ratio of 2\% for external and 4\% for internal datasets (ablation in Appendix.~\ref{sec:num-abla}). The SPRT detector uses the count-based LLR estimation (Eq.~\ref{eq:cnt_llr}) with both type-I/II error budgets $\alpha,\beta$ set to 0.1. Scalability is tested by extending internal memory by 200 records (Appendix.~\ref{sec:scalab-exp}). All experiments run on a single NVIDIA RTX 5090 GPU.

\textbf{Training Configuration.} We train a transformer-based attacker proxy on 256 samples per attack, augmented via similarity top-k random walks to simulate exploratory trajectories. Honeypots are optimized to attract this proxy while maintaining contrastive separation from 500 generated benign queries. To ensure scalability, we employ balanced k-means clustering to partition the corpus, generating honeypot embeddings for each cluster via a shared network to minimize computational cost.

\textbf{Attacks and Baselines.} We evaluate against RAG-Thief~\cite{jiang2024rag}, DGEA~\cite{cohen2024unleashing}, and IKE~\cite{wang2025IKEA} for external memory; and MEXTRA~\cite{wang2025unveiling} plus an adapted IKEA for internal memory. Baselines include ControlNet~\cite{yao2025controlnet} (100 anchor queries) and LLM-based Agent Detectors (Appendix.~\ref{sec:defender-details})~\cite{zhang2024intention,zeng2024autodefense}. Additionally, we implement a theoretical Optimal Sequential Detector (Optimal-Seq) to empirically validate \cref{thm:adv_stat short}. The details of all setups are shown in Appendix.~\ref{sec:setups-detail}.

\subsection{Evaluation Metrics}\label{sec:metrics}
We assess the performance from two perspectives:

\textbf{Detection Effectiveness.} We measure the ability to distinguish attackers from normal users using the Area Under the Receiver Operating Characteristics Curve (AUROC). To evaluate performance under strict service quality constraints, we specifically report the True Positive Rate at low False Positive thresholds (TPR@1\%FPR and TPR@10\%FPR).

\textbf{Detection Efficiency.} 
We evaluate the computational overhead and detection speed:
(1) Delay (s): The average additional online inference latency per turn introduced by the defense.
(2) First Detection Time (FDT): The average number of interaction turns (samples $N$ in SPRT) required to make a detection decision. Lower FDT indicates earlier interception of leakage.

\subsection{Performance Evaluation with Existing Baselines}
We conducted 64-round attacks on both domains. As shown in \cref{tab:external_performance_main}, \emph{MemPot} maintains $>0.96$ AUROC on external memory, while baselines collapse against stealthy IKEA attacks ($\approx 0.50$ AUROC). This performance gap widens in internal settings (\cref{tab:internal_performance_main}): against MEXTRA and the human-mimicking IKEA, baselines yield near-random results (e.g., 0.40 AUROC on EHRAgent), whereas \emph{MemPot} achieves near-perfect accuracy with zero online latency by effectively aggregating sequential evidence.

\begin{table}[!ht]
\centering
\caption{Utility impact of the \emph{MemPot} defense on standard task performance.}
\label{tab:utility_performance}
\resizebox{\columnwidth}{!}{
\begin{tabular}{l cc ccc}
\toprule
\multirow{2}{*}{Setting} & \multicolumn{2}{c}{RAP WebShop} & \multicolumn{3}{c}{Pokémon QA \& MCQ} \\
\cmidrule(lr){2-3} \cmidrule(lr){4-6}
 & Score & Success Rate & Acc & Rouge-L & Sim \\
\midrule
w/o pots & \textbf{67.1} & 44.6 & \textbf{0.98} & \textbf{0.67} & \textbf{0.75} \\
w/ pots & 65.4 & 43.8 & \textbf{0.98} & 0.66 & \textbf{0.75} \\
\bottomrule
\end{tabular}
}
\vskip -1.5em
\end{table}

\subsection{Comparison with Optimal Static Detector}
To empirically validate \cref{thm:adv_stat short}, we compare \emph{MemPot} against the Optimal Sequential Detector (Optimal-Seq). As shown in \cref{tab:health_pokemon_opt} and \cref{tab:ehr_rap_opt}, Optimal-Seq struggles against stealthy attacks due to the inherent indistinguishability of the original distribution (e.g., 0.70 AUROC, 22-round delay against MEXTRA-Cosine). In contrast, \emph{MemPot} actively reshapes the retrieval landscape to amplify adversarial signals, achieving near-perfect performance (0.99 AUROC) and reducing detection time to just 2 rounds. This confirms that proactive distribution modification is essential to break the limits of static detection.

\subsection{Utility Impact of \emph{MemPot}}
To verify utility preservation, we evaluate \emph{MemPot} on standard tasks on benign user traces. For internal memory dependent tasks, we evaluate Score and Success Rate (defined in RAP~\cite{RAP_kagaya_24}) on WebShop benchmark~\cite{webshop_yao_neurips22}. 
For external memory dependent tasks, we evaluate Accuracy (Acc), Similarity (Sim) and Rouge-L scores on QA and MCQ tasks, with the same setting of IKEA~\cite{wang2025IKEA}. As shown in~\cref{tab:utility_performance},
\emph{MemPot} has negligible impact on standard agent capabilities, with the Success Rate on the WebShop dropping only marginally from 44.6 to 43.8 and Pokémon QA metrics remain almost unchanged, confirming that our optimized honeypots are non-disruptive to benign user interactions. We also provide several examples of honeypot documents to show that they are harmless and cause no negative impact, even when retrieved by benign users (Appendix.~\ref{sec:pot-example}).

\subsection{Analysis of Different LLR Estimation Methods}
We compare three LLR estimation methods (\cref{sec:pot_detect}) on $\text{MEXTRA}_\text{Cosine}$ (\cref{tab:llr_estimation}). The Count-based method offers the best trade-off ($>0.99$ AUROC, FDT $\approx$ 2) and is adopted as our default. While Global-similarity achieves the fastest detection at the cost of FPR stability, the Similarity-based method proves suboptimal on complex tasks (e.g., 0.74 AUROC on EHRAgent).

\section{Conclusion}
We propose \emph{MemPot}, an active defense framework that safeguards agent memory by reshaping retrieval distributions with honeypots to amplify adversarial signals. Overcoming the theoretical limits of passive detection, \emph{MemPot} achieves near-perfect accuracy against advanced threats like IKEA and MEXTRA with zero online latency and preserving agent utility.

\newpage








\section*{Impact Statement}
This paper presents work aimed at advancing the security and privacy of Large Language Model (LLM) agents. By mitigating the risks of knowledge extraction in Retrieval-Augmented Generation (RAG) and agentic memory systems, our framework contributes to the protection of proprietary intellectual property and sensitive user data, particularly in high-stakes domains such as healthcare and finance. We believe this work fosters the trustworthy deployment of autonomous agents by ensuring that memory capabilities do not become vulnerabilities. We do not foresee immediate negative societal consequences, as our active defense mechanism is designed to distinguish malicious probing from legitimate usage, minimizing the risk of disrupting normal service.


\bibliography{reference}
\bibliographystyle{icml2026}

\newpage
\appendix
\onecolumn

\section{Details of Honeypot optimization}\label{sec:detail-pots-training}
\noindent\textbf{Attacker Proxy. }
In practical settings, the detailed algorithms and objectives of attackers are often not accessible, making it difficult to obtain an exact attack query distribution $\mathcal{Q}_1$. To address this, we construct an \textit{attacker proxy} that learns to reproduce the attacker’s behavior at the sentence embedding level. The proxy is realized as a neural network \(A_{\omega}\) that, given the historical interaction trajectory \(\{o_{1:t-1}\}\)—including previous query embeddings, retrieved document indices, and similarity scores—iteratively outputs the embedding of the next query \(\hat{e}_{t} = A_{\omega}(o_{1:t-1})\). 
Instead of reconstructing the attacker’s internal algorithm, the goal of \(A_{\omega}\) is to mimic the observable embedding-level dynamics of the attacker. The training objective minimizes the discrepancy between the generated embeddings and the observed ones using cosine similarity regression, ensuring that the proxy reproduces similar retrieval rankings under the same retriever. This design allows the honeypot system to train defensive strategies without explicit access to the attacker’s model parameters.

\noindent\textbf{User Proxy. }
For normal user modeling, we face a similar challenge: the real-world user query distribution is difficult to collect and often highly diverse. To approximate it, we leverage the large language model’s (LLM) human-simulator capability to synthesize realistic user behaviors. Specifically, we utilize a human-simulation pipeline that prompts the LLM with interaction intents and retrieval contexts, generating natural multi-turn query sequences that resemble genuine information-seeking behavior. These synthetic user queries are then encoded into embeddings to approximate the benign query distribution $\mathcal{Q}_0$.

\section{Details of Experiment Setups}
\subsection{Setups}
\label{sec:setups-detail}
\textbf{Datasets and Evaluation Setting.} 
We evaluate \emph{MemPot} across four diverse benchmarks covering both external memory retrieval and internal memory interaction scenarios:
For external memory, we utilize the HealthMagicCare (20k rows)~\cite{LavitaHealth} and Pokémon (9.46k rows)~\cite{asoria2024pokemon} datasets, representing privacy-intensive healthcare inquiries and entity-heavy gaming knowledge, respectively. We default retrieve 4 documents each turn in experiments. For internal memory, following the setting in MEXTRA~\cite{wang2025unveiling}, we use the generated internal memory logs for EHRAgent~\cite{EHRAgent_shi_emnlp24} and RAP on WebShop~\cite{webshop_yao_neurips22, RAP_kagaya_24}. The standard setting contains 300 interaction records. We follow the same retrieval number setting (top-4 records for EHRAgent and top-3 for RAP WebShop) in origin paper~\cite{EHRAgent_shi_emnlp24, RAP_kagaya_24}. For scalability analysis (Appendix.~\ref{sec:scalab-exp}), we extend this by appending an additional 200 records. We employ DeepSeek-v3.2 (685B) as the backbone model of both agents and all-mpnet-base-v2 as the retrieval embedding model. We use mDeBERTa-v3-base~\cite{he2021debertav3} as NLI model. For each dataset/agent, we use the same honeypot set against different attacks. We inject honeypots at a ratio of 2\% for external memory datasets and 4\% for internal memory datasets by default (We evaluate the impact of honeypots ratio in Appendix.~\ref{sec:num-abla}). For the SPRT detector, we utilize the Pot-NonPot Counts Ratio (Eq.~\ref{eq:cnt_llr}) as the default method for LLR estimation with both type-I/II error budgets $\alpha,\beta$ set to 0.1. All experiments are conducted on a single NVIDIA RTX 5090 GPU.

\textbf{Training Set Configuration.} 
We train a transformer-based attacker proxy using 256 samples for each attack type. To robustly capture the exploratory nature of extraction attacks, we augment the training data using a similarity top-k random walk, which simulates an attacker's trajectory by iteratively sampling the next query from the semantic neighbors of the current retrieval. The honeypots are trained by interacting with this proxy, while simultaneously optimizing contrastive loss against 500 LLM-generated benign human queries for each target dataset. To support honeypots' training on large-scale knowledge bases, we first apply balanced k-means clustering to partition the document corpus into equal-sized sets. The document embeddings within each set are fed into a shared network to generate a corresponding honeypot embedding, significantly reducing the computational burden of generation.

\textbf{Attacks and Defense Baselines.} 
We consider the following attacks: 
For external memory, we evaluate RAG-targeted adaptive attacks including RAG-Thief~\cite{jiang2024rag}, DGEA~\cite{cohen2024unleashing}, and IKEA~\cite{wang2025IKEA}. For internal memory, we evaluate agent-targeted MEXTRA~\cite{wang2025unveiling}. Additionally, we adapt IKEA for internal memory extraction to assess robustness against adaptive attacks.
We broadly consider existing defense methods, including ControlNet~\cite{yao2025controlnet} and Agent Detector~\cite{zhang2024intention,zeng2024autodefense, agarwal2024prompt}. Specifically, we use 100 benign anchor queries for ControlNet initiation and DeepSeek-v3.2 (685B) as the Agent Detector backbone (see Appendix~\ref{sec:defender-details}).
We additionally design a transformer-based theoretical optimal sequential detector (Optimal-Seq) for further comparison and validation of \cref{thm:adv_stat short}. 

\subsection{Agent Detector Setting} \label{sec:defender-details}
Referring to mitigation suggestions in ~\citep{zeng2024good,jiang2024rag,anderson2024my, zhang2024intention, zeng2024autodefense}, We apply the agent detector with hybrid paradigms, including intention detection, keyword detection and defensive instruction. Specifically, we use DeepSeek-v3.2 (685B) as the agent detector backbone. The response generation process integrated with the detector is shown as follows:
For an input query $q$, defense first occurs through intent detection~\citep{zhang2024intention} and keyword filtering~\citep{zeng2024good}:
\begin{equation}
    q_{\text{defended}} = 
\begin{cases} 
    \emptyset, &  D_{\text{intent}}(q) \lor D_{\text{keyword}}(q) = 1 \\ 
    q, & \text{otherwise}
    \end{cases},
\end{equation}
where $\emptyset$ enforces an \text{``unanswerable''} response, $D_{\text{intent}}(\cdot)$ and $D_{\text{keyword}}(\cdot)$ are detection functions which return True when detecting malicious extraction intention or words. 
When $q_{\text{defended}} \neq \emptyset$, generation combines the retrieval context $\mathcal{D}^{K}_q$ is:
\begin{equation}
    y = \text{LLM}\big(\textrm{Concat}(\mathcal{D}^{K}_q) \oplus q_{\text{defended}} 
    \oplus p_{\text{defense}}\big
    ),
\end{equation}
where defensive prompt $p_{\text{defense}}$~\citep{agarwal2024prompt} constrains output relevance by prompting LLM only answer with related part of retrievals, and enforces LLM not responding to malicious instruction with provided examples.

\section{Additional Experiments}
\subsection{Scalability of MemPot under Memory Updates} 
\label{sec:scalab-exp}
To address the dynamic nature of agent systems where memory is updated in real-time, we evaluate the scalability of \emph{MemPot} using a batch-based update strategy. Leveraging the localized nature of retrieval, we hypothesize that honeypots can be optimized independently for different memory batches and directly merged into a unified defense set. 

\textbf{Setup.} We simulate a memory update scenario by appending 200 new interaction records (sourced from the MEXTRA dataset~\cite{wang2025unveiling}) to the standard 300-record internal memory, resulting in a total of 500 records. We compare two implementation strategies: (1) \textbf{MemPot-E2E}: The computationally expensive upper bound, where honeypots are re-optimized globally on the complete 500-record dataset from scratch. (2) \textbf{MemPot-Stack}: The scalable approach, where we retain the honeypots for the initial 300 records and simply append a new set of honeypots optimized specifically for the 200 update records. To ensure a fair comparison, we maintain the honeypot ratio at 4\% for both strategies.

\textbf{Results.} As shown in~\cref{tab:stack-agent}, \emph{MemPot-Stack} achieves detection performance nearly identical to the holistic \emph{MemPot-e2e} across all attack vectors. Notably, in challenging scenarios like MEXTRA-General on EHRAgent, the stacked approach maintains superior robustness (0.98 AUROC) compared to the end-to-end baseline (0.94 AUROC). These results confirm that \emph{MemPot} supports efficient, modular updates: as the agent's memory grows, new honeypots can be seamlessly integrated without the need for global retraining, ensuring continuous protection with minimal computational overhead.
\begin{table*}[t]
\caption{Performance evaluation of stacked MemPot and end-to-end trained MemPot on EHRAgent and Web-Shopping RAP agent.}
\label{tab:stack-agent}
\centering
\footnotesize
\begin{tabular}{c l ccccc ccccc}
\toprule
\multirow{4}{*}{Attack} & \multicolumn{1}{c}{\multirow{4}{*}{Defense}} & \multicolumn{5}{c}{EHRAgent} & \multicolumn{5}{c}{RAP-web} \\
\cmidrule(lr){3-7} \cmidrule(lr){8-12}
 & & AUROC & \makecell{TP@\\1\%FP} & \makecell{TP@\\10\%FP} & Delay & FDT & AUROC & \makecell{TP@\\1\%FP} & \makecell{TP@\\10\%FP} & Delay & FDT \\
\midrule
\multirow{3}{*}{$\text{MEXTRA}_\text{Cosine}$} 
 & Optimal-Seq & 0.70 & 0.16 & 0.27 & 0.03 & 22 & 0.68 & 0.06 & 0.22 & 0.04 & 38 \\
 & MemPot-E2E & \textbf{0.99} & \textbf{0.67} & 0.74 & \textbf{0} & \textbf{2} & 0.98 & \textbf{1.00} & \textbf{1.00} & \textbf{0} & \textbf{1} \\
 & MemPot-Stack & 0.98 & 0.65 & \textbf{0.75} & \textbf{0} & 3 & \textbf{0.99} & \textbf{1.00} & \textbf{1.00} & \textbf{0} & \textbf{1} \\
\midrule
\multirow{3}{*}{$\text{MEXTRA}_\text{Edit}$} 
 & Optimal-Seq & 0.71 & 0.14 & 0.32 & 0.04 & 33 & 0.75 & 0.20 & 0.36 & 0.04 & 27 \\
 & MemPot-E2E & \textbf{0.98} & 0.56 & \textbf{0.74} & \textbf{0} & \textbf{4} & \textbf{1.00} & \textbf{1.00} & \textbf{1.00} & \textbf{0} & \textbf{1} \\
 & MemPot-Stack & 0.97 & \textbf{0.58} & 0.72 & \textbf{0} & 9 & \textbf{1.00} & \textbf{1.00} & \textbf{1.00} & \textbf{0} & \textbf{1} \\
\midrule
\multirow{3}{*}{$\text{MEXTRA}_\text{General}$} 
 & Optimal-Seq & 0.69 & 0.16 & 0.28 & 0.04 & 42 & 0.73 & 0.12 & 0.26 & 0.04 & 35 \\
 & MemPot-E2E & 0.94 & 0.81 & 0.85 & \textbf{0} & 7 & \textbf{1.00} & \textbf{1.00} & \textbf{1.00} & \textbf{0} & \textbf{1} \\
 & MemPot-Stack & \textbf{0.98} & \textbf{0.82} & \textbf{0.88} & \textbf{0} & \textbf{5} & \textbf{1.00} & \textbf{1.00} & \textbf{1.00} & \textbf{0} & \textbf{1} \\
\bottomrule
\end{tabular}

\end{table*}

\subsection{Impact of Honeypots Number}\label{sec:num-abla}
We investigate the trade-off between ratio of injected honeypots and detection robustness. \cref{tab:ablation_pots} reports the detection performance on the RAP WebShop with 300 memory records as the honeypots ratio increases from 0.6\% to 4\%. While the AUROC remains saturated at 1.00 even with minimal injection, increasing the honeypot count significantly enhances detection speed and sensitivity at strict thresholds. Specifically, increasing pots from 0.6\% to 4\% improves the TP@1\%FP from 0.77 to 0.94 and reduces the First Detection Time (FDT) from 7 rounds to just 2 rounds. This trend validates that a denser honeypot distribution amplifies the adversarial signal, allowing for faster interception of attacks, though a small budget (e.g., 2\%) already yields near-optimal performance.
\begin{table}[t]
\footnotesize
\caption{Ablation study on the number of pots (honeypots) and their impact on detection performance on $\text{MEXTRA}_\text{cosine}$ in RAP WebShop.}
\label{tab:ablation_pots}
\centering
\begin{tabular}{lcccc}
\toprule
\multicolumn{1}{c}{Ratio} & AUROC & TP@1\%FP & TP@10\%FP & FDT \\
\midrule
0.6\% & \textbf{1.00} & 0.77 & \textbf{1.00} & 7 \\
1.0\% & \textbf{1.00} & 0.79 & \textbf{1.00} & 4 \\
2.0\% & \textbf{1.00} & 0.90 & \textbf{1.00} & \textbf{2} \\
4.0\% & \textbf{1.00} & \textbf{0.94} & \textbf{1.00} & \textbf{2} \\
\bottomrule
\end{tabular}
\end{table}

\section{Theoretical Preliminary}
\label{sec:prelim}
\begin{definition}[Fixed parametric partition]
Let \(T_\phi:\mathcal{O}\to\mathcal{Y}\) be a fixed-form parametric partition (its \emph{form} does not change during training; only parameters \(\phi\) change). Define the push-forwards
\[
P_{1,\theta,\phi}=T_{\phi\#}f_{1,\theta},\qquad
P_{0,\theta,\phi}=T_{\phi\#}f_{0,\theta},
\]
and the divergence
\(
\Phi_T(\theta,\phi)=\KL(P_{1,\theta,\phi}\|P_{0,\theta,\phi}).
\)
\end{definition}

\begin{lemma}[Data Processing Inequality (DPI)]\label{lem:DPI}
For any measurable \(T_\phi\),
\[
\KL(f_{1,\theta}\|f_{0,\theta})\ \ge\ \KL(T_{\phi\#}f_{1,\theta}\|T_{\phi\#}f_{0,\theta})
=\Phi_T(\theta,\phi).
\]
Likewise,
\(
\KL(f_{0,\theta}\|f_{1,\theta})\ge \KL(P_{0,\theta,\phi}\|P_{1,\theta,\phi}).
\)
\end{lemma}

\begin{proof}
Classic DPI; see \emph{Step 1} of Lemma~\ref{lem:supremum}'s \emph{proof}.
\end{proof}

\begin{lemma}[Conditional expectation over finite partition]\label{lem:CE-partition}
Let \((\Omega,\mathcal F,g)\) be a probability space, let \(\mathcal P=\{A_1,\dots,A_m\}\) be a finite
measurable partition of \(\Omega\), and write \(\sigma(\mathcal P)\) for the \(\sigma\)-algebra it generates.
For any \(X\in L^1(g)\), define
\[
Y(x):=\sum_{i=1}^m \mathbf 1_{A_i}(x)\,
\begin{cases}
\displaystyle \frac{1}{g(A_i)}\int_{A_i} X\,dg, & g(A_i)>0,\\[1.5ex]
0, & g(A_i)=0.
\end{cases}
\]
Then:
\begin{enumerate}
\item \(Y\) is \(\sigma(\mathcal P)\)-measurable and, for every \(B\in\sigma(\mathcal P)\),
\(\displaystyle \int_B Y\,dg=\int_B X\,dg\). Hence \(Y=\E_g[X\mid \sigma(\mathcal P)]\) almost surely.
\item In particular, writing \(\E_g[X\mid\mathcal P]\) as shorthand for \(\E_g[X\mid\sigma(\mathcal P)]\),
\begin{equation}\label{eq:tower-partition}
\int \E_g[X\mid \mathcal P]\,dg \;=\; \int X\,dg .
\end{equation}
\end{enumerate}
\end{lemma}

\begin{proof}
(1) By construction, \(Y\) is constant on each atom \(A_i\), thus \(\sigma(\mathcal P)\)-measurable.  
If \(B=\bigcup_{i\in I} A_i\in\sigma(\mathcal P)\), then
\begin{align*}
    \int_B Y\,dg 
    &= \sum_{i\in I}\frac{1}{g(A_i)}\!\left(\int_{A_i} X\,dg\right) g(A_i)\\
    &= \sum_{i\in I}\int_{A_i} X\,dg = \int_B X\,dg,
\end{align*}

which is precisely the defining property of the conditional expectation
\(\E_g[X\mid \sigma(\mathcal P)]\). Uniqueness up to \(g\)-null sets yields \(Y=\E_g[X\mid \sigma(\mathcal P)]\) a.s.

(2) Take \(B=\Omega\) in the identity of part (1) to obtain \eqref{eq:tower-partition}.
\end{proof}

\begin{lemma}[Partition supremum]\label{lem:supremum}
For any pair of laws \(f,g\) on \(\mathcal{O}\),
\[
\KL(f\|g)=\sup_{T}\ \KL(T_{\#}f\|T_{\#}g),
\]
where the supremum is over all finite measurable partitions (equivalently, finite-range measurable maps \(T\)); moreover, for any \(\delta>0\) there exists a finite partition \(T_\delta\) such that
\[
\KL(f\|g)\le \KL(T_{\delta\#}f\|T_{\delta\#}g)+\delta.
\]
\end{lemma}

\begin{proof}
If \(f\not\ll g\), there exists \(A\in\mathcal F\) with \(g(A)=0\) and \(f(A)>0\); for the two-atom
partition \(\{A,A^c\}\) one has \(\KL(T_{\#}f\|T_{\#}g)=+\infty=\KL(f\|g)\), and the statement is trivial.
Hence assume \(f\ll g\). Let \(L:=\frac{df}{dg}\) and \(\ell:=\log L\),  which are well-defined with Radon-Nikodym theorem. Then
\(\KL(f\|g)=\int L\log L\,dg=\int \ell\,df\in[0,\infty]\).

\vskip 1em
\noindent \emph{Step 1 (DPI Inequality).}

\noindent Let \(\mathcal P=\{A_i\}_{i=1}^m\) be a finite partition and let \(T\) be its index map.
Write \(p_i=f(A_i)=\int_{A_i}L\,dg\) and \(q_i=g(A_i)\) (with the convention \(0\log 0:=0\)).
Then
\begin{align*}
    \KL(T_{\#}f\|T_{\#}g)
    &=\sum_{i=1}^m p_i\log\!\frac{p_i}{q_i}\\
    &=\sum_{i=1}^m \Big(\int_{A_i}L\,dg\Big)\log
    \frac{\int_{A_i}L\,dg}{g(A_i)}\\
    &=\int \, \E_g[L\mid\mathcal P]\;\log \E_g[L\mid\mathcal P]\; dg,
\end{align*}
where \(\E_g[\cdot\mid\mathcal P]\) is conditional expectation under \(g\) onto the
\(\sigma\)-algebra generated by \(\mathcal P\).
Since \(\varphi(u):=u\log u\) is convex on \((0,\infty)\), Jensen yields
\(\varphi(\E_g[L\mid\mathcal P])\le \E_g[\varphi(L)\mid\mathcal P]\), and with Lemma~\ref{lem:CE-partition} integrating gives
\begin{align*}
    \KL(T_{\#}f\|T_{\#}g)
    &=\int \varphi(\E_g[L\mid\mathcal P])\,dg\\
    &\;\le\; \int \E_g[\varphi(L)\mid\mathcal P]\,dg\\
    &=\int \varphi(L)\,dg=\KL(f\|g).
\end{align*}

\noindent Then taking the supremum over all finite \(\mathcal P\) shows
\(\sup_T \KL(T_{\#}f\|T_{\#}g)\le \KL(f\|g)\).

\vskip 1em
\noindent \emph{Step 2 (Limitation of divergence gap with a finite partition).}

\noindent Fix \(\delta>0\) and assume \(\KL(f\|g)<\infty\).
We construct a finite partition \(\mathcal P_{\delta}\) for which
\(\KL(f\|g)-\KL(T_{\delta\#}f\|T_{\delta\#}g)\le\delta\).

\smallskip
\noindent \emph{(2.a) Tail control.}

Choose \(M\ge 1\) so large that the ``upper tail'' contribution satisfies
\[
\int_{\{L\ge e^{M}\}} L\log L\,dg \;\le\; \delta/3,
\]
and additionally \(2Me^{-M}\le \delta/3\) (possible since \(Me^{-M}\to 0\) as \(M\to\infty\)).
Define three regions
\[
A_-:=\{L<e^{-M}\},
B:=\{e^{-M}\le L<e^{M}\},
A_+:=\{L\ge e^{M}\}.
\]

\smallskip
\noindent \emph{(2.b) Middle quantization.}

Pick a mesh size \(\eta\in(0,1)\) to be specified (below we take \(\eta:=\delta/3\)).
Partition the middle region \(B\) into finitely many level sets of \(L\):
for \(k=0,1,\dots,K-1\) with \(K:=\lceil 2M/\eta\rceil\), set
\[
B_k:=\big\{\, x\in\mathcal O:\ e^{-M+k\eta}\le L(x)< e^{-M+(k+1)\eta}\,\big\}.
\]
Then on each \(B_k\) we have \( \log L \in [a_k,a_k+\eta]\) with \(a_k:=-M+k\eta\), and also

\begin{align*}
&\log \E_g[L\mid B_k]=\log\frac{\int_{B_k}L\,dg}{g(B_k)}\in [a_k,a_k+\eta]\\
&\quad \Rightarrow
\big|\log L - \log \E_g[L\mid B_k]\big|\le \eta\ \text{ on }B_k.
\end{align*}

\smallskip
\noindent \emph{(2.c) The finite partition.}

Let
\(\mathcal P_\delta:=\{A_-,\,B_0,\dots,B_{K-1},\,A_+\}\)
and let \(T_\delta\) be its index map.
Using the identity from Step~1 and writing the gap as an \(L\,dg\)-integral, we have
\begin{align*}
    \KL(f\|g)-\KL(T_{\delta\#}f\|T_{\delta\#}g)
    =\sum_{C\in\mathcal P_\delta} \int_{C} L\Big(\log L-\log \E_g[L\mid C]\Big)\,dg
    =: \Delta_-\ +\ \Delta_B\ +\ \Delta_+ .
\end{align*}

\smallskip
\noindent \emph{(2.d) Bounding the middle gap.}

On each \(B_k\), the pointwise bound \(|\log L-\log \E_g[L\mid B_k]|\le \eta\) yields
\begin{align*}
    \Delta_B
    =\sum_{k=0}^{K-1} \int_{B_k} L\big(\log L-\log \E_g[L\mid B_k]\big)\,dg
    \;\le\; \eta \sum_{k=0}^{K-1}\int_{B_k} L\,dg
    \;=\; \eta\, f(B)
    \;\le\; \eta.
\end{align*}

\smallskip
\noindent \emph{(2.e) Bounding the lower-tail gap.}

On \(A_-\) one has \(L\le e^{-M}\).
Using \(|\log L-\log \E_g[L\mid A_-]|\le |\log L|+|\log \E_g[L\mid A_-]|\) and
\(\E_g[L\mid A_-]\le e^{-M}\Rightarrow |\log \E_g[L\mid A_-]|\le -M\), we obtain
\begin{align*}
    \Delta_-
\le \int_{A_-} L|\log L|\,dg + M\!\int_{A_-} L\,dg
\le M e^{-M} + M e^{-M}
= 2Me^{-M}
\;\le\; \delta/3,
\end{align*}
by the choice of \(M\).

\smallskip
\noindent \emph{(2.f) Bounding the upper-tail gap.}
On \(A_+\) one has \(\E_g[L\mid A_+]\ge e^{M}\), so
\(\log L-\log \E_g[L\mid A_+]\le \log L - M\).
Hence
\begin{align*}
    0\le \Delta_+
=\int_{A_+} L\big(\log L-\log \E_g[L\mid A_+]\big)\,dg
\le \int_{A_+} L(\log L-M)\,dg
\le \int_{A_+} L\log L\,dg
\;\le\; \delta/3,
\end{align*}
by the choice of \(M\).

\smallskip
\noindent \emph{(2.g) Total control.}

Set \(\eta:=\delta/3\).
Collecting the bounds from (2.d)--(2.f) gives
\[
\KL(f\|g)-\KL(T_{\delta\#}f\|T_{\delta\#}g)
\le \delta/3 + \delta/3 + \delta/3=\delta.
\]

\vskip 1em
\noindent \emph{Step 3 (Taking the supremum).}

\noindent By Step~1, \(\KL(T_{\#}f\|T_{\#}g)\le \KL(f\|g)\) for every finite partition.
By Step~2, for each \(\delta>0\) there exists a finite partition \(T_\delta\) such that
\(\KL(f\|g)\le \KL(T_{\delta\#}f\|T_{\delta\#}g)+\delta\).
Therefore \(\sup_T \KL(T_{\#}f\|T_{\#}g)\ge \KL(f\|g)-\delta\) for all \(\delta>0\),
hence \(\sup_T \KL(T_{\#}f\|T_{\#}g)=\KL(f\|g)\).
\end{proof}

\begin{lemma}[Cross-entropy dominates the Bayes risk]
\label{lem:Cross-entropy dominates}
Let $J\in\{1,\dots,K\}$ and $\mathbf Y=(Y_1,\ldots,Y_K)$ be generated as in the
$1$-positive $(K{-}1)$-negative scheme, and let $\pi^*(j\mid\mathbf Y)$ denote the
true posterior of $J$ given $\mathbf Y$. For any measurable score
$h:\mathcal Y\to\mathbb R$, define the model posterior
$\pi_h(j\mid\mathbf Y):=\frac{e^{h(Y_j)}}{\sum_{i=1}^K e^{h(Y_i)}}$ and the
$K$-sample InfoNCE loss
\[
\mathcal L_{\mathrm{NCE},K}(h):=\E\!\left[-\log \pi_h(J\mid\mathbf Y)\right].
\]
Then
\begin{equation}\label{eq:ce-bayes}
\begin{aligned}
    \mathcal L_{\mathrm{NCE},K}(h)\ \ge\ 
    \E_{\mathbf Y}\!\big[H(\pi^*(\cdot\mid\mathbf Y))\big]
    \quad\Longleftrightarrow\quad
    -\mathcal L_{\mathrm{NCE},K}(h)\ \le\ -\mathcal L^*,
\end{aligned}
\end{equation}
where $H(p):=-\sum_j p(j)\log p(j)$ and
$\mathcal L^*:=\E_{\mathbf Y}\big[H(\pi^*(\cdot\mid\mathbf Y))\big]$ is the Bayes risk when taking logarithmic loss. Equality holds iff $\pi_h(\cdot\mid\mathbf Y)=\pi^*(\cdot\mid\mathbf Y)$ a.s.
\end{lemma}

\begin{proof}
For fixed $\mathbf Y$, by the standard decomposition $H(p,q)=H(p)+\KL(p\|q)$, the cross-entropy between $\pi^*$ and $\pi_h$ is
\begin{align*}
    H(\pi^*,\pi_h)
    =\E_{J\sim \pi^*(\cdot\mid\mathbf Y)}[-\log \pi_h(J\mid\mathbf Y)]
    =H(\pi^*)+\KL\!\big(\pi^*(\cdot\mid\mathbf Y)\,\|\,\pi_h(\cdot\mid\mathbf Y)\big).
\end{align*}
Since
$\KL(\cdot\|\cdot)\ge 0$ (Gibbs' inequality~\cite{cover2012infotheory}), we have
$H(\pi^*,\pi_h)\ge H(\pi^*)$ with equality iff $\pi_h=\pi^*$.
Taking expectation over $\mathbf Y$ gives
\[
\mathcal L_{\mathrm{NCE},K}(h)
=\E_{\mathbf Y}\big[H(\pi^*,\pi_h)\big]
\ge \E_{\mathbf Y}\big[H(\pi^*)\big]=\mathcal L^*,
\]
which is \eqref{eq:ce-bayes}. 
\end{proof}

\begin{lemma}[Conditional product density under a uniform index]\label{lem:cond-product-density}
Let $P,Q$ be probability laws on $(\mathcal Y,\mathcal G)$ with $P\ll Q$, and set
$r:=\frac{dP}{dQ}$. Fix $K\ge2$. Draw $J\sim\mathrm{Unif}\{1,\dots,K\}$ and, given $J=j$,
sample $\mathbf Y=(Y_1,\ldots,Y_K)$ with independent coordinates
$Y_j\sim P$ and $Y_i\sim Q$ for $i\neq j$. Then:
\begin{enumerate}
\item For each $j$, the conditional law $\mathbb P(\mathbf Y\mid J=j)$ is absolutely continuous
w.r.t.\ $Q^{\otimes K}$ with Radon--Nikodym derivative
\[
\frac{d\,\mathbb P(\mathbf Y\mid J=j)}{d\,Q^{\otimes K}}(\mathbf y)\;=\; r(y_j)
\qquad(Q^{\otimes K}\text{-a.e.}).
\]
\item The marginal law of $\mathbf Y$ satisfies
\[
\frac{d\,\mathbb P_{\mathbf Y}}{d\,Q^{\otimes K}}(\mathbf y)\;=\;\frac1K\sum_{i=1}^K r(y_i)
\qquad(Q^{\otimes K}\text{-a.e.}).
\]
\end{enumerate}
\end{lemma}

\begin{proof}
We use the test-function characterization of Radon--Nikodym derivatives.

\smallskip
\noindent \emph{(1) Conditional law.}
Fix $j\in\{1,\dots,K\}$ and any bounded measurable $\varphi:\mathcal Y^K\to\mathbb R$.
By the construction of $\mathbf Y$ given $J=j$ and independence of coordinates,
\[
\int \varphi(\mathbf y)\, d\mathbb P(\mathbf Y\mid J=j)
= \int \varphi(\mathbf y)\, d\big(Q^{\otimes (j-1)}\otimes P\otimes Q^{\otimes (K-j)}\big)(\mathbf y).
\]
Since $P\ll Q$ with density $r=\frac{dP}{dQ}$ and $Q\ll Q$ with density $1$, we have for
$Q^{\otimes K}$ almost everywhere on\ $\mathbf y$,
\[
\frac{d\big(Q^{\otimes (j-1)}\otimes P\otimes Q^{\otimes (K-j)}\big)}{d(Q^{\otimes K})}(\mathbf y)
= r(y_j)\cdot \prod_{i\ne j}1
= r(y_j),
\]
and therefore
\[
\int \varphi(\mathbf y)\, d\mathbb P(\mathbf Y\mid J=j)
= \int \varphi(\mathbf y)\, r(y_j)\, d(Q^{\otimes K})(\mathbf y).
\]
By uniqueness in the Radon--Nikodym theorem, this identifies
$\frac{d\,\mathbb P(\mathbf Y\mid J=j)}{d\,Q^{\otimes K}}=r(y_j)$ almost everywhere.

\smallskip
\noindent \emph{(2) Marginal law.}
Averaging over the uniform $J$,
\begin{align*}
    \int \varphi(\mathbf y)\, d\mathbb P_{\mathbf Y}
= \sum_{j=1}^K \mathbf{1}(J=j)\int \varphi(\mathbf y)\, d\mathbb P(\mathbf Y\mid J=j)
= \frac1K\sum_{j=1}^K \int \varphi(\mathbf y)\, r(y_j)\, d(Q^{\otimes K})(\mathbf y).
\end{align*}
Hence, again by Radon--Nikodym uniqueness,
\(
\dfrac{d\,\mathbb P_{\mathbf Y}}{d\,Q^{\otimes K}}(\mathbf y)
= \dfrac1K\sum_{j=1}^K r(y_j)
\)
$Q^{\otimes K}$.
\end{proof}

\begin{lemma}[ASN Approximation for SPRT with Markov Observations]
\label{lem:ASN_Markov}
Consider the SPRT with boundaries \(A<0<B\) and error budgets \((\alpha,\beta)\). 
Let the observation sequence \(\{O_t\}_{t=1}^\infty\) be a stationary and ergodic Markov chain under each hypothesis, with transition densities \(f_{1,\theta}\) and \(f_{0,\theta}\).  
Let \(\ell_\theta(O_t|O_{t-1}) = \log \frac{f_{1,\theta}(O_t|O_{t-1})}{f_{0,\theta}(O_t|O_{t-1})}\) be the per‑step log‑likelihood ratio.
Define the drift rates under \(H_1\) and \(H_0\) respectively as
\[
\mu_1(\theta) = \mathbb{E}_{1,\theta}[\ell_\theta(O_t|O_{t-1})], \qquad
\mu_0(\theta) = \mathbb{E}_{0,\theta}[\ell_\theta(O_t|O_{t-1})].
\]
If the overshoot at stopping is negligible, then the expected stopping times (ASN) are approximately
\begin{equation}\label{eq:ASN-exp}
    \mathbb{E}_{1}[N] \approx \frac{|\log B|}{\mu_1(\theta)}, \qquad 
    \mathbb{E}_{0}[N] \approx \frac{|\log A|}{|\mu_0(\theta)|}. 
\end{equation}
\end{lemma}

\begin{proof}
We give a detailed proof under the assumption of stationary, ergodic Markov observations.

Let \(\{O_t\}_{t\ge 0}\) be a Markov chain on a state space \(\mathcal{O}\). 
Under hypothesis \(H_i\) (\(i=0,1\)), the chain has transition density \(f_{i,\theta}(\cdot|O_{t-1})\) and a unique stationary distribution \(\pi_i\).  
We assume the chain starts from its stationary distribution (or an arbitrary initial distribution; by ergodicity the long‑run behavior is the same).  
The per‑step log‑likelihood ratio is
\[
Z_t \equiv \ell_\theta(O_t|O_{t-1}) = \log \frac{f_{1,\theta}(O_t|O_{t-1})}{f_{0,\theta}(O_t|O_{t-1})},
\]
and the cumulative sum is \(S_n = \sum_{t=1}^n Z_t\).  
The stopping time is
\[
N = \inf\{ n \ge 1 : S_n \ge \log B \text{ or } S_n \le \log A \},
\]
where the boundaries satisfy \(A = \frac{\beta}{1-\alpha}\) and \(B = \frac{1-\beta}{\alpha}\) for given error probabilities \(\alpha,\beta\) (Wald’s approximations).

By stationarity,
\[
\mu_1(\theta) = \mathbb{E}_{1,\theta}[Z_t] 
= \iint \pi_1(dx) f_{1,\theta}(dy|x) \log\frac{f_{1,\theta}(dy|x)}{f_{0,\theta}(dy|x)}
= \mathrm{KL}\big(f_{1,\theta}(\cdot|O)\|f_{0,\theta}(\cdot|O)\big),
\]
where the last equality denotes the average Kullback–Leibler divergence under the stationary distribution \(\pi_1\).  
Similarly,
\[
\mu_0(\theta) = \mathbb{E}_{0,\theta}[Z_t] 
= -\mathrm{KL}\big(f_{0,\theta}(\cdot|O)\|f_{1,\theta}(\cdot|O)\big) < 0.
\]

Since the chain is ergodic under each hypothesis, the strong law of large numbers for Markov chains gives
\[
\frac{S_n}{n} \xrightarrow[n\to\infty]{\text{a.s.}} \mu_1(\theta) \quad \text{under } H_1,
\]
and
\[
\frac{S_n}{n} \xrightarrow[n\to\infty]{\text{a.s.}} \mu_0(\theta) \quad \text{under } H_0.
\]

Under \(H_1\), the process \(S_n\) grows approximately linearly with drift \(\mu_1(\theta)\).  
Ignoring the overshoot when \(S_N\) first crosses a boundary, the time to reach the upper boundary \(\log B\) is roughly
\[
N \approx \frac{\log B}{\mu_1(\theta)} .
\]
Taking expectations on both sides yields \(\mathbb{E}_1[N] \approx \frac{|\log B|}{\mu_1(\theta)}\) (note \(\log B>0\)).

Under \(H_0\), the process drifts downward with slope \(\mu_0(\theta)<0\).  
The time to hit the lower boundary \(\log A\) (with \(\log A<0\)) is approximately
\[
N \approx \frac{\log A}{\mu_0(\theta)} = \frac{|\log A|}{|\mu_0(\theta)|},
\]
hence \(\mathbb{E}_0[N] \approx \frac{|\log A|}{|\mu_0(\theta)|}\).

Therefore, under the conditions of stationarity, ergodicity, and negligible overshoot, the expected sample numbers are well approximated by the expressions in \eqref{eq:ASN-exp}.
\end{proof}

\begin{remark}
The drift rates \(\mu_1(\theta)\) and \(\mu_0(\theta)\) are exactly the stationary Kullback–Leibler divergences between the transition laws.  
For i.i.d. observations, the Markov dependence vanishes and the lemma reduces to Wald’s classical ASN formulas.  
In practice, the approximation is accurate when the boundaries are sufficiently far from the starting point (i.e., when \(\alpha\) and \(\beta\) are small).
\end{remark}

\section{Proof of InfoNCE upper-bound}
\label{pf:infoNCE}
\begin{theorem}[InfoNCE upper-bound by information drift]
\label{thm:infoNCE upb}
Let $P:=T_{\phi\#}f_{1,\theta}$ and $Q:=T_{\phi\#}f_{0,\theta}$ on $\mathcal Y$.
Consider the sampling scheme: draw $J\sim\mathrm{Unif}\{1,\dots,K\}$, then
$Y_J\sim P$ and $(Y_i)_{i\neq J}\stackrel{\text{i.i.d.}}{\sim}Q$,
independently of $J$. For any score $h:\mathcal Y\to\mathbb R$, define the
InfoNCE loss
\[
\mathcal L_{\mathrm{NCE},K}(h)
:=-\,\E\!\left[\log
\frac{e^{h(Y_J)}}{\sum_{i=1}^K e^{h(Y_i)}}
\right].
\]
Then, for every $K\ge2$,
\begin{equation}\label{eq:nce-kl}
-\mathcal L_{\mathrm{NCE},K}(h)\ \le\ \KL(P\|Q)-\log(K).
\end{equation}
\end{theorem}

\begin{proof}
Let $r(y):=\frac{dP}{dQ}(y)$ and write $\mathbf Y:=(Y_1,\dots,Y_K)$.
By Bayes' rule, the posterior distribution of $J$ given $\mathbf Y$ is
\begin{equation}
\pi^*(j\mid\mathbf Y)
==\frac{\mathbb P(J=j,\mathbf Y)}{\sum_{i=1}^K \mathbb P(J=i,\mathbf Y)}
=\frac{r(Y_j)}{\sum_{i=1}^K r(Y_i)}.
\end{equation}
The Bayes-optimal multiclass log-loss is
\begin{equation}
    \mathcal L^*:=\E\!\left[-\log \pi^*(J\mid\mathbf Y)\right]
    =\E\!\left[\log\sum_{i=1}^K r(Y_i)\right]-\E[\log r(Y_J)].
\end{equation}

With Lemma~\ref{lem:Cross-entropy dominates}, the cross-entropy with any model is no smaller than the Bayes risk, which means
\begin{equation}
    \label{eq:bayes_loss-nce_loss}
    \mathcal L_{\mathrm{NCE},K}(h)\ge \mathcal L^*
\end{equation}
holds for all $h$, and by using Bayes' rule and the uniform prior $\mathbb P(J=j)=1/K$, we have
\begin{equation}
    \label{eq:loss-mutual_info}
    \begin{aligned}
        I(J;Y)
        &= \mathbb E\!\left[\log\frac{\mathbb P(J\mid Y)}{\mathbb P(J)}\right]
        = \mathbb E[\log \pi^*(J\mid Y)] - \log(1/K)
        = \log K - \mathbb E[-\log \pi^*(J\mid Y)]
        = \log K - \mathcal L^*
    \end{aligned}.
\end{equation}
We have the marginal with Lemma~\ref{lem:cond-product-density}:
\begin{equation}\label{eq:marginal-density}
\frac{d\mathbb P_{\mathbf Y}}{dQ^{\otimes K}}(\mathbf y)=\frac1K\sum_{i=1}^K r(y_i).
\end{equation}
With $J$ uniform,
\begin{equation}\label{eq:I-decomp}
    I(J;\mathbf Y)
    = \E\!\left[\log\frac{\pi^*(J\mid\mathbf Y)}{1/K}\right]
    = \E[\log r(Y_J)] - \E\!\left[\log \sum_{i=1}^K r(Y_i)\right] + \log K. 
\end{equation}
The first term equals the KL divergence because $Y_J\sim P$ marginally:
\begin{equation}\label{eq:first-term}
\E[\log r(Y_J)] = \E_{P}[\log \frac{dP}{dQ}(y)] 
=\KL(P\|Q).
\end{equation}
For the second term, non-negativity of KL yields
\begin{equation}\label{eq:second-term}
    \begin{aligned}
        0\le \KL(\mathbb P_{\mathbf Y}\,\|\,Q^{\otimes K})
        =\E_{\mathbf Y}\!\left[\log \frac{d\mathbb P_{\mathbf Y}}{dQ^{\otimes K}}(\mathbf Y)\right]
        =\E\!\left[\log\frac1K\sum_{i=1}^K r(Y_i)\right]
        =\E\!\left[\log\sum_{i=1}^K r(Y_i)\right]-\log K,
    \end{aligned}
\end{equation}
which indicates\ $\log K-\E[\log\sum_i r(Y_i)]\le 0$.
Then substituting \eqref{eq:first-term} and \eqref{eq:second-term} into \eqref{eq:I-decomp} gives
$I(J;\mathbf Y)\le \KL(P\|Q)$.

\smallskip
\noindent Finally, by \eqref{eq:bayes_loss-nce_loss} and \eqref{eq:loss-mutual_info},
\[
-\mathcal L_{\mathrm{NCE},K}(h)\ \le\ I(J;\mathbf Y)-\log K\ \le\ \KL(P\|Q)-\log K.
\]
\end{proof}

\begin{corollary}[Loss to projected KL]\label{cor:loss2phi}
For fixed-form \(T_\phi\), decreasing \(\mathcal{L}_{\mathrm{NCE}}(\theta,\phi)\) increases a valid lower bound to \(\Phi_T(\theta,\phi)\).
\end{corollary}

\section{Proof of advantage over static test}
\label{pf:adv_stat}
\begin{theorem}[Decreasing Loss leads to larger drifts and smaller expected samples]\label{thm:goal}
Fix error budgets \((\alpha,\beta)\). Consider the SPRT based on the true LLR increments \(\ell_\theta\).
Let \((\theta^\star,\phi^\star)\) minimize \(\mathcal{L}_{\mathrm{NCE}}(\theta,\phi)\) over an admissible set. Then:
\begin{enumerate}
\item \(\Phi_T(\theta^\star,\phi^\star)\ge \Phi_T(\theta_0,\phi^\star)\), hence by DPI
\[
\KL(f_{1,\theta^\star}\|f_{0,\theta^\star})\ \ge\ \Phi_T(\theta^\star,\phi^\star)\ \ge\ \Phi_T(\theta_0,\phi^\star).
\]
\item If, in addition, the statistic family \(T_\phi\) is rich enough that for every \(\delta>0\) there exists \(\phi=\phi(\delta)\) with
\(
\Phi_T(\theta_0,\phi)\ \ge\ \KL(f_{1,\theta_0}\|f_{0,\theta_0})-\delta,
\)
then for the minimizer \((\theta^\star,\phi^\star)\) we have
\[
\mu_1(\theta^\star)\ \ge\ \mu_1(\theta_0) - \delta,
\]
and analogously for \(|\mu_0|\). In particular, for arbitrarily small \(\delta\), the drifts are (weakly) increased.
\item By Wald's relations \eqref{eq:wald-approx} and SPRT optimality, the expected sample sizes $E_1[N]$ and $E_0[N]$ (at the same \((\alpha,\beta)\)) are (weakly) \emph{decreased}.
\end{enumerate}
\end{theorem}

\begin{proof}
(1) follows from Theorem~\ref{thm:infoNCE upb} and Corollary~\ref{cor:loss2phi}. For (2) apply Lemma~\ref{lem:supremum} at \(\theta_0\): for any \(\delta>0\) there exists a (finite) partition \(T_\delta\) such that
\(
\KL(f_{1,\theta_0}\|f_{0,\theta_0})\le \KL(T_{\delta\#}f_{1,\theta_0}\|T_{\delta\#}f_{0,\theta_0})+\delta.
\)
If the parametric family \(T_\phi\) is dense w.r.t.\ this partition topology (e.g.\ neural universal approximation in $L^1$; see Remark~\ref{rem:approx}), we can pick \(\phi^\star\) with \(\Phi_T(\theta_0,\phi^\star)\ge \KL(T_{\delta\#}f_{1,\theta_0}\|T_{\delta\#}f_{0,\theta_0})-\epsilon\), and let \(\epsilon\to0\). Then
\[
\mu_1(\theta^\star)\ \ge\ \Phi_T(\theta^\star,\phi^\star)\ \ge\ \Phi_T(\theta_0,\phi^\star)\ \ge\ \KL(f_{1,\theta_0}\|f_{0,\theta_0})-\delta.
\]
The statement for \(|\mu_0|\) is identical by swapping roles of \(H_0,H_1\). Finally (3) follows from \eqref{eq:wald-approx} and the SPRT optimality of expected sample size at the given \((\alpha,\beta)\) \cite{wald1992sequential}.
\end{proof}

\begin{remark}[On separating generator and statistic parameters]\label{rem:theta-phi}
It is often convenient to write \(\Phi(\theta,\phi)=\KL(T_{\phi\#}f_{1,\theta}\|T_{\phi\#}f_{0,\theta})\) with \(\theta\) controlling the index (thus the observation laws) and \(\phi\) controlling the statistic. In implementations one may tie them (\(\phi=\theta\)); the proof above treats \((\theta,\phi)\) jointly and only requires the \emph{form} of \(T_\phi\) to be fixed.
\end{remark}

\begin{remark}[Approximation richness]\label{rem:approx}
Universal approximation~\cite{cybenko1989approximation} results imply that parametric families of measurable maps (neural networks, piecewise-constant partitions, histogram features) are dense in \(L^1\) on compact domains. Together with lower semicontinuity of \(f\)-divergences \cite{sason2015bounds}, this justifies the  $\delta$-tightening in Theorem~\ref{thm:infoNCE upb}.
\end{remark}

\begin{lemma}[SPRT is optimal in expected length at same \((\alpha,\beta)\)]\label{lem:sprt_optim}
For any \(\theta\) and any competing (possibly fixed-length) test achieving \((\alpha,\beta)\), the SPRT (with the same \((\alpha,\beta)\)) satisfies
\[
E_1[N]_{\mathrm{SPRT},\theta}\ \le\ E_1[N]_{\mathrm{any},\theta},
E_0[N]_{\mathrm{SPRT},\theta}\ \le\ E_0[N]_{\mathrm{any},\theta}.
\]
In particular \(E_b[N]_{\mathrm{SPRT},\theta}\le n^\star_{\mathrm{static}}(\alpha,\beta;\theta)\) for \(b\in\{0,1\}\).
\end{lemma}

\begin{proof}
This is the classical optimality of the SPRT due to Wald and Wolfowitz (see \cite{wald1948optimum}).
\end{proof}

\begin{lemma}[Existence of $\theta_0$ and a.s.\ equivalence to static]\label{lem:theta0_exist}
Suppose the following assumptions hold and $N\ge K$:
\begin{enumerate}
    \item Scoring uses $s_{\cos}$ and $\Phi$ is a deterministic top-$K$ map. Any tie events have zero $Q_b$-probability and are resolved by a fixed rule.
    \item For $b\in\{0,1\}$, $\mathbb P_{q\sim Q_b}\!\left(c_K(q)>0\right)=1$.
    \item Let $\bar Q_b$ be the law of $\hat q$ under $H_b$, and
    $\bar Q:=\tfrac12(\bar Q_0+\bar Q_1)$.
    The support $\operatorname{supp}(\bar Q)$ is contained in some closed hemisphere of $\mathbb S^{d-1}$.
    Equivalently, the polar set is nonempty:
    \[
    \mathcal U_{\cos}\ :=\ \Big\{\hat u\in\mathbb S^{d-1}:\sup_{v\in \operatorname{supp}(\bar Q)}\langle \hat u, v\rangle\le 0\Big\}
    \]
\end{enumerate}

\noindent Then there exists $\theta_0$ with
$E_{\mathrm{pot}}(\theta_0)=\{u_j\}_{j=1}^P$ satisfying $\hat u_j\in\mathcal U_{\cos}$ for all $j$,
such that for $b\in\{0,1\}$,
\[
\mathbb P_{q\sim Q_b}\!\Big(\ \max_{1\le j\le P}\langle \hat q,\hat u_j\rangle \;<\; c_K(q)\ \Big)=1
\Longrightarrow
f_{b,\theta_0}=f^{\mathrm{static}}_b,
\]
i.e., under $\theta_0$ the honeypots almost surely never enter the top-$K$ and the
induced observation law equals the \emph{static} (documents-only) law.
\begin{proof}
If $\hat u_j\in\mathcal U_{\cos}$ then for any $\hat q\in\operatorname{supp}(\bar Q)$,
$\langle \hat q,\hat u_j\rangle\le 0$; 
by assumption, $c_K(q)>0$ almost surely.
Hence $\max_j\langle \hat q,\hat u_j\rangle\le 0<c_K(q)$ almost surely,
so top-$K$ coincides with the documents-only top-$K$ almost surely.
Deterministic $\Phi$ implies equality of the induced laws. 
\end{proof}
\end{lemma}

\begin{theorem}[Loss-minimizing training improves the operational goal over the static optimal]\label{thm:vs-static}
There exist untrained parameter \(\theta_0\) make vector database equivalent to vector database without honeypots and \(\theta^\star\) the post-training parameter obtained by minimizing \(\mathcal{L}_{\mathrm{NCE}}\). Under the assumptions of Theorem~\ref{thm:goal}, for the same \((\alpha,\beta)\), and \(b\in\{0,1\}\),
\[
E_b[N]_{\mathrm{SPRT},\theta^\star}\ \le\ \min\!\Big\{E_b[N]_{\mathrm{SPRT},\theta_0},\ E_b[N]_{\mathrm{any},\text{static}}\},
\]
with strict inequality whenever the drifts increase strictly.
\end{theorem}

\begin{proof}
By Lemma~\ref{lem:theta0_exist}, there exists \(\theta_0\) under assumptions such that \(f_{b,\theta_0}=f^{\mathrm{static}}_b\) for \(b\in\{0,1\}\). Hence the SPRT at \((\alpha,\beta)\) under \(\theta_0\) has expected sample size equal to the static optimal with Lemma~\ref{lem:sprt_optim}.
By Theorem~\ref{thm:goal}, the drifts \(\mu_1(\theta),|\mu_0(\theta)|\) (weakly) increase after training, hence the Wald relations \eqref{eq:wald-approx} imply (weakly) smaller $E_b[N]$ for the SPRT at \((\alpha,\beta)\), and $\mu_b(\theta) \ge \mu_b(\theta_0)$ holds by letting $\delta\rightarrow 0$.
Combine all have the inequality.
\end{proof}

\newpage
\newcommand{\RETURN}{\STATE \textbf{return}~}
\section{Algorithm of Scorer Guided embedding inversion}
\label{emb_inv}
Scorer guided beam-search embedding inversion algorithm~\cite{zhang2025advInv} is shown below:
\begin{algorithm}[ht]
\caption{Embedding Inversion}
\label{alg:adv_decoding}
\begin{algorithmic}[1]
\renewcommand{\algorithmicrequire}{\textbf{Hyperparameters:}}
\renewcommand{\algorithmicinput}{\textbf{Input:}}
\renewcommand{\algorithmicensure}{\textbf{Output:}}

\REQUIRE beam width $b$, top-$k$
\INPUT prefix prompt $P$, target embedding $\mathbf{e}_{\text{target}}$
\ENSURE best found sequence of length $\textit{max\_length}$

\STATE \textbf{Initialize:} Beams $\mathcal{B} = \{\text{empty string}\}$

\FOR{each time step $t$ from $1$ to $\textit{max\_length}$}
    \STATE $\mathcal{B}_{\text{new}} \leftarrow \{\}$
    \STATE $\mathcal{S}_{\text{new}} \leftarrow \{\}$
    \FOR{each beam $b \in \mathcal{B}$}
        \STATE $z_t \leftarrow \text{LLM}_{\text{logits}}(P \oplus b)$
        \STATE $\text{topk\_tokens} \leftarrow \text{TopK}(z_t, k)$
        \FOR{each token $t_k \in \text{topk\_tokens}$}
            \STATE $b' \leftarrow b \oplus t_k$
            \STATE $\mathcal{B}_{\text{new}}.\text{append}(b')$
            \STATE $\mathcal{S}_{\text{new}} \leftarrow \text{Scorer}(\mathcal{B}_{\text{new}}, \mathbf{e}_{\text{target}},S,\mathcal{E}_{\text{doc}})$
        \ENDFOR
    \ENDFOR
    \STATE \text{Sort } $\mathcal{B}_{\text{new}}$ \text{ by } $\mathcal{S}_{\text{new}}$
    \STATE $\mathcal{B} \leftarrow \mathcal{B}_{\text{new}}[:b]$
\ENDFOR
\RETURN $\mathcal{B}[0]$
\end{algorithmic}
\end{algorithm}

\section{Examples of Honeypot Documents}\label{sec:pot-example}
We provide a few honeypot documents example here, corresponding to WebShop, Pokémon and Healthcare scenarioes respectively:

\begin{tcolorbox}[colframe=orange!85!black, title=Honeypot Documents Example]
    \textcolor{orange}{\textbf{Pot 1:}} ``Web Shopping: **1432**:  With just six keywords like search term in 'search: [ ]'; there already comes together several possibilities within three of which all work or more (inclusu[.]m any other product to find them for sure without searching again from memory alone as this will get out''\\

    \textcolor{orange}{\textbf{Pot 2:}} ``Alright, Nature offers an enchanting and unique array of creatures with incredible strength! A rare breed from Japan - Pokémon that has the physical properties described below: * Name – P.E''\\

    \textcolor{orange}{\textbf{Pot 3:}} ``Creating concise passages from this resource while generating appropriate dialogue within my parameters where all data needs filtering every fourth point  but those facts/assums or "sugar-crate-loads The patient reports experiencing various physical discomforts and functional changes.''
    
\end{tcolorbox}

\end{document}